\documentclass[11pt]{article}
\usepackage[left=1in,top=1in,right=1in,bottom=1in,head=0in,nofoot]{geometry}

\usepackage{setspace,url,bm,amsmath} 
\usepackage{scalerel}

\usepackage{xcolor}
 \usepackage{multirow}
 \usepackage{arydshln}
 \usepackage{mathtools}
\usepackage{titlesec} 
\titlelabel{\thetitle.\quad} 
\titleformat*{\section}{\bf\Large\center}

\usepackage{authblk}
\usepackage{float}
\usepackage{graphicx} 
\usepackage{bbm}
\usepackage{latexsym}
\usepackage{caption}
\usepackage[margin=20pt]{subcaption}
\usepackage{enumerate}
\graphicspath{ {./Graphs/} }

\usepackage{amsthm}
\usepackage{amssymb}
\usepackage{amsmath}
\usepackage{color}

\usepackage{comment}
\theoremstyle{definition}
\newtheorem{assumption}{Assumption}
\newtheorem*{theorem*}{Theorem}
\newtheorem{theorem}{Theorem}
\newtheorem*{rmk*}{remark}
\newtheorem{proposition}{Proposition}
\newtheorem{lemma}{Lemma}
\newtheorem{example}{Example}

\newtheorem{condition}{Condition}

\newtheorem{corollary}{Corollary}
\newtheorem*{corollary*}{Corollary}

\usepackage{natbib} 
\bibpunct{(}{)}{;}{a}{}{,} 

\usepackage{etoolbox} 
\apptocmd{\sloppy}{\hbadness 10000\relax}{}{} 

\usepackage{multibib}
\newcites{sec}{References}

\usepackage{color}
\usepackage{listings}
\usepackage{hyperref}
\usepackage{booktabs}

\usepackage{lscape}

\RequirePackage[normalem]{ulem}

\def \E  {\mathbb{E}}
 
\def \P{\mathbb{P}} 
\def \V {\mathbb{V}}
\def \bfE {\mathbb{E}}

\newcommand{\indep}{\perp \!\!\! \perp}

\allowdisplaybreaks

\begin{document}

\singlespacing

\title{\bf  
Quantifying Individual Risk for Binary Outcomes} 

\author[1]{Peng Wu}
\author[2]{Peng Ding}
\author[1]{Zhi Geng}
\author[3]{Yue Liu\thanks{Corresponding author: liuyue\_stats@ruc.edu.cn. This article is to appear in JRSSB.}}
\affil[1]{\small School of Mathematics and Statistics, Beijing Technology and Business University, 100048, China}
\affil[2]{\small Department of Statistics, University of Californi, Berkeley, CA 94720, USA}  
\affil[3]{\small Center for Applied Statistics and School of Statistics, Renmin University of China, 100872,  China}  



\date{}

\maketitle

\begin{abstract}
Understanding treatment effect heterogeneity is crucial for reliable decision-making in treatment evaluation and selection. The conditional average treatment effect (CATE) is widely used to capture treatment effect heterogeneity induced by observed covariates and to design individualized treatment policies. However, it is an average metric within subpopulations, which prevents it from revealing individual risk, potentially leading to misleading results. This article fills this gap by examining individual risk for binary outcomes, specifically focusing on the fraction negatively affected (FNA), a metric that quantifies the percentage of individuals experiencing worse outcomes under treatment compared with control.  Even under the strong ignorability assumption, FNA is still unidentifiable, and the existing Fr\'{e}chet--Hoeffding bounds are often too wide and attainable only under extreme data-generating processes. By invoking mild conditions on the value range of the Pearson correlation coefficient between potential outcomes, we obtain improved bounds compared with the Fr\'{e}chet--Hoeffding bounds. We show that paradoxically,  even with a positive CATE, the lower bound on FNA can be positive, i.e., in the best-case scenario, many individuals will be harmed if they receive treatment. Additionally, we establish a nonparametric sensitivity analysis framework for FNA using the Pearson correlation coefficient as the sensitivity parameter. Furthermore, we propose nonparametric estimators for the refined FNA bounds and prove their consistency and asymptotic normality. We use simulation to evaluate the performance of the proposed estimators and apply the method to 
a canonical observational study.
\end{abstract}

\medskip 
\noindent 
{\bf Keywords}: 
Causal Inference; Fraction Negatively Affected; Partial Identification; Sensitivity Analysis


\newpage

\onehalfspacing

\section{Introduction}  \label{sec1}
Understanding treatment effect heterogeneity is essential for reliable decision-making in treatment evaluation and selection~\citep{murphy2003optimal, Chakraborty-Moodie2013, Kosorok+Laber:2019, chernozhukov-etal2022}.
 Generally, heterogeneity in individual treatment effects consists of two components: (i) systematic heterogeneity explained by observed covariates, and (ii) idiosyncratic heterogeneity not explained by observed covariates~\citep{Heckman-etal1997, Djebbari-Smith2008, Ding-etal2019}. 
Existing work utilized the conditional average treatment effect (CATE) to capture treatment effect heterogeneity and design tailored treatment policies~\citep[see e.g.,][]{Imai-Strauss2011, Kent-etal2018, 2018Who, 2021Policy}.

However, as an average metric in subpopulations defined by covariate values, CATE reflects only the systematic heterogeneity. It fails to capture all individual heterogeneity,   
  which is critical for reliable decision-making~\citep{Bolger-etal2019, Lei-Candes2021, Jin-etal2023, Wu-Mao2025}. Consequently, CATE is inadequate for revealing individual  risk associated with treatment in subpopulations and, in some cases, can be misleading~\citep{2022nathan, Eli-etal2022, Mueller-Pearl2023, MuellerPearl2023-CausalInference}. For instance, consider a study with two subpopulations. 
 In one subpopulation, a drug benefits 80\% but harms the remaining 20\%. For the other subpopulation, the drug benefits 60\%  while the rest are unaffected, experiencing neither benefit nor harm. Although the CATE is 0.6 in both subpopulations, recommendations for administering the drug to the first subpopulation should be reconsidered due to the potential harm. 

In this article, we explore individual risk arising from the full  individual heterogeneity.  We investigate the fraction negatively affected (FNA) by treatment for binary outcomes~\citep{2022nathan}. Also known as the treatment harm rate~\citep{Zhang-etal2013}, 
FNA quantifies the percentage of individuals experiencing worse outcomes when receiving treatment compared with control. 
It serves as a metric for assessing individual risk over the population. 
Existing work discussed the identifiability of FNA under assumptions for potential outcomes, like monotonicity~\citep{Huang-etal2012}, conditional independence~\citep{Shen-etal2013}, and latent conditional independence with parametric model restrictions~\citep{Yin-etal2018}.  However, these assumptions are difficult to meet in practice.  Given that FNA involves the joint distribution of  the potential outcomes whereas we observe only their marginals,  FNA is generally unidentifiable, even in randomized controlled trials (RCTs). 

Rather than aiming to identify FNA directly, several studies focus on establishing its bounds. \cite{Gadbury2004} derived bounds on FNA in RCTs without covariates.    
 In observational studies with unconfoundedness, \cite{Zhang-etal2013} improved the bounds by incorporating covariates and applying the Fr\'{e}chet--Hoeffding  inequality within each stratum defined by covariates. \cite{Yin2018-second} employed a secondary outcome to achieve tighter bounds. 
\cite{2022nathan} discussed the bounds of FNA induced by a treatment policy. 


We demonstrate that while the 
 Fr\'{e}chet--Hoeffding bounds on FNA are sharp, they are attainable only under extreme data-generating processes.
When the data-generating process is not extreme, the upper bound (the worst-case) becomes overly conservative, and the lower bound (the best-case) becomes overly optimistic.  
 Additionally, we observe that the lower bound always remains zero when  CATE exceeds zero. As the treatment is typically assigned only when CATE is positive, the lower bound provides no information on FNA. Consequently, improving the bounds is essential for ensuring more informative and trustworthy decision-making.

   We make three main contributions. First, we obtain improved bounds compared with previous studies by introducing a mild assumption on the range of the Pearson correlation coefficient (PCC) between potential outcomes. This assumption is easy to integrate with domain knowledge, as experienced practitioners can provide rough ranges for the PCC.  
In addition, we find that paradoxically, even with a positive CATE, the lower bound on FNA can be positive when the PCC is small, i.e., in the best-case scenario, many individuals will be harmed by treatment.  
The improved bounds encompass several previous (partial) identifiability results in certain degenerate situations, including studies by \cite{Shen-etal2013}, \cite{Zhang-etal2013}, and \cite{2022nathan}. 
Moreover, the improved lower and upper bounds delineate the best and worst cases for FNA, carrying practical implications for decision-making, policy evaluation, and treatment policy learning.
Second, we present a nonparametric sensitivity analysis framework, utilizing the PCC as the sensitivity parameter. This framework naturally arises from bounding FNA under assumptions on the PCC. 
Third, we propose nonparametric estimators for the bounds on FNA  by employing the  semiparametric theory~\citep{Bickel-etal1993}, and establish their consistency and asymptotic normality.
We validate our findings through extensive simulations and real-world applications.  
The remainder of this paper is organized as follows. Section \ref{sec2} describes the problem of interest and the Fr\'{e}chet--Hoeffding bounds on FNA. Section \ref{sec3} introduces a novel sensitivity analysis method for FNA given covariates, presents the improved bounds, and  illustrates its practical usage. 
 Section \ref{sec5} extends the proposed method to the marginal FNA, 
 proposes estimators for the bounds on FNA, and shows their large sample properties. Section \ref{sec7} assesses the finite-sample performance of the proposed method through a simulation study. Section \ref{sec8} demonstrates our methods with an empirical example. Section \ref{sec9} concludes with a discussion. 
The Supplementary Material contains additional results and technical details.
We provide the replication materials at \url{https://github.com/pengwu1224/Quantifying-Individual-Risk-for-Binary-Outcome}.


\section{Setup} \label{sec2}

\subsection{Notation} \label{sec2-1}
We adopt the potential outcomes framework to define causal effects. Let 
$\{ (X_i, A_i, Y_i^1, Y_i^0): i=1, ..., n\}$ represent an independent and identically distributed sample from a superpopulation $\P$. For each unit $i$,  $X_i \in \mathcal{X}\subset \mathbb{R}^p$ is a vector of pre-treatment covariates, $A_i \in \mathcal{A} = \{0,  1\}$ is a binary treatment where $A_i=1$ indicates receiving treatment and $A_i=0$  indicates receiving control, and $Y_i^a \in \mathcal{Y} = \{0, 1\}$ is the potential outcome if the unit were to receive the treatment $a$ for $a = 0, 1$. 
Since each unit is assigned to either the treatment or control, the observed outcome equals $Y_i = (1-A_i)Y_i^0 + A_i Y_i^1$.  
Let $Y_i=1$ represent a favorable outcome (e.g., survival) and $Y_i=0$ an unfavorable outcome (e.g., death). 

Define the individual treatment effect (ITE) as $\tau_i = Y_i^1-Y_i^0$, with $\tau_i>0$ indicating that receiving treatment is beneficial for unit $i$, and vice versa.
Define the conditional average treatment effect (CATE)  as $ \tau(x) = \bfE[\tau_i \mid X_i =x] = \bfE[ Y_i^1 - Y_i^0 \mid X_i =x],$  
	which represents 
the difference in the conditional mean outcomes between receiving treatment and control. The average treatment effect (ATE) is $\tau = \bfE[Y_i^1 -Y_i^0]$.   
Throughout, we drop the subscript $i$ for a generic unit, and maintain the commonly-used strong ignorability assumption below. 
\begin{assumption} \label{assumption-1} 
(a) Unconfoundedness: $(Y^0 ,Y^1)\indep A \mid X $;   
(b) Overlap: $\epsilon <  e(x) := \P(A=1 \mid X=x)<1- \epsilon$, where $0< \epsilon < 1/2$ is a constant. 
\end{assumption}


 Assumption \ref{assumption-1}(a) requires that $X$ includes all confounders affecting both the outcome and treatment.   Assumption \ref{assumption-1}(b) requires that units with any given $X$ have a positive probability of receiving treatment. 
 Under Assumption \ref{assumption-1},
 $$\tau(x) =  \mu_1(x) - \mu_0(x), \text{ where } \mu_a(x) = \bfE[Y \mid X = x, A = a], \quad a = 0, 1.$$  

Generally, $\tau(x)$ is the most fine-grained identifiable average treatment effect. 
 However, $\tau(x)$ is a subpopulation-level metric, specifically for the subpopulation characterized by $X=x$, rather than at the individual level~\citep{pearl2016}. It overlooks inherent individual variability within the subpopulation, which may be critical for individualized decision-making~\citep{Lei-Candes2021}. To illustrate, we provide the following example to illustrate its limitation.

\begin{example} \label{ex-tmp} Consider a case where covariate $X$, treatment $A$, and outcome $Y$ are all binary variables. 
The treatment \( A \) indicates whether a drug is administered (\( A=1 \)) or not (\( A=0 \)), and \( Y=0 \) and \( Y=1 \) denote death and survival, respectively.  
Suppose that for subgroup with $X=0$, the drug benefits 80\% (\(Y^1 - Y^0 = 1\)) but harms the remaining 20\% (\(Y^1 - Y^0 = -1\)). For subgroup with $X=1$, the drug benefits 60\% (\(Y^1 - Y^0 = 1\)), while the rest are unaffected (\(Y^1 - Y^0 = 0\)). Although the CATEs are the same in both subpopulations (\(\tau(0) = \tau(1) = 0.6\)), recommendations for administering the drug to subgroup with $X=0$ should be carefully considered due to the potential harm, despite the positive CATE.  
\end{example}

Example \ref{ex-tmp} shows that $\tau(x)$ cannot fully capture individual risk. To fill this gap, we introduce several estimands to evaluate individual risk.

\subsection{Causal Estimands}
We present causal parameters to quantify individual risk. Similar to \cite{2022nathan, Kallus2023}, if the ITE $\tau_i < 0$, we say the unit $i$ is negatively affected by treatment. Then the average fraction negatively affected (FNA) is defined as $\text{FNA} = \P(Y^0 = 1, Y^1 = 0)$, 
which measures the average individual risk over the whole population.
 Likewise, we define FNA given $X=x$ as 
  $$\text{FNA}(x) = \P(Y^0 = 1, Y^1 = 0\mid X=x).    $$  
 Similarly, one can define the fraction positively affected (FPA) as 
		$\text{FPA} = \P( Y^0 = 0, Y^1 = 1)$, the FPA given $X=x$ as $\text{FPA}(x) = \P( Y^0 = 0, Y^1 = 1 \mid X=x)$.  
We can verify that $\tau(x) = \text{FPA}(x) - \text{FNA}(x)$,  suggesting  that $\text{FPA}(x)$ and $\text{FNA}(x)$ offer a more informative characterization of individual risk than $\tau(x)$. 
		
It also is of interest to explore the average individual risk induced by a specific treatment policy. Let $d: \mathcal{X}\to \mathcal{A}$ be a treatment policy that assigns a treatment to each unit based on its covariates.  
Define FNA induced by the policy $d$ as   
$\text{FNA}(d) = \bfE[ \P(Y^0 = 1, Y^1 = 0 \mid X)d(X)] = \bfE[ \text{FNA}(X) \cdot d(X) ]$,  
 and FPA induced by the policy $d$ as $\text{FPA}(d) = \bfE [\text{FPA}(X) \cdot d(X)]$.    
In this article, we focus on $\text{FNA}$ and $\text{FNA}(x)$. Other causal estimands can be addressed similarly. 

\subsection{Fr\'{e}chet--Hoeffding Bounds on FNA}  \label{sec2-3}
Identifying $\text{FNA}$ generally requires stringent assumptions as it involves the joint distribution of the potential outcomes. Notably, even randomization does not guarantee its identifiability,  which poses challenges in assessing individual risk. 
Typically, when an estimand is not identifiable, the focus shifts toward determining its bounds. Previous studies have obtained the sharp Fr\'{e}chet--Hoeffding bounds on $\mathrm{FNA}$~\citep{Zhang-etal2013, 2022nathan}, which are summarized in Lemma \ref{lem-1} below. 

\begin{lemma}[Fr\'{e}chet--Hoeffding Bounds] \label{lem-1} Assume Assumption \ref{assumption-1}.

(a) $\mathrm{FNA}(x) \in [ L_{\mathrm{FNA}}(x),  U_{\mathrm{FNA}}(x)],$
where   $L_{\mathrm{FNA}}(x)   
	   					= \max \big \{  \mu_0(x) - \mu_1(x),  0 \big  \},$ and 
	 $U_{\mathrm{FNA}}(x)  
	 				= \min \big  \{ \mu_0(x),  1 - \mu_1(x) \big  \}$. 
These bounds on $\mathrm{FNA}(x)$ are sharp. 

(b) $\mathrm{FNA} \in [ L_{\mathrm{FNA}},  U_{\mathrm{FNA}}],$
where $ L_{\mathrm{FNA}}  = \bfE \big [  L_{\mathrm{FNA}}(X)  \big ]$ and $ U_{\mathrm{FNA}}   =  \bfE \big [ U_{\mathrm{FNA}}(X) \big ].$ 
These bounds on $\mathrm{FNA}$ are sharp. 
\end{lemma}

While the bounds in Lemma \ref{lem-1} are sharp, they are attainable only under extreme data-generating processes. 
The upper bound is attained if and only if either $\P(Y^0=1, Y^1 = 1 \mid X=x) = 0$ or $\P(Y^0=0, Y^1 = 0 \mid  X=x) = 0$, while the lower bound is attained if and only if  $\P(Y^0=0, Y^1 = 1 \mid  X=x) = 0$ or $\P(Y^0=1, Y^1 = 0 \mid  X=x) = 0$. 
  When these conditions do not hold, the bounds in Lemma \ref{lem-1} provide loose information about $\text{FNA}(x)$. 
 In addition, randomization cannot improve the bounds in Lemma \ref{lem-1} as it does not offer any information about the joint distribution of the potential outcomes.   
Given the importance of $\text{FNA}(x)$, it is valuable to obtain its improved bounds. 
As indicated in Lemma \ref{lem-1}(b), the bounds on $\text{FNA}$ can always be derived by taking expectations of the bounds on $\text{FNA}(X)$ over $X$. 
In what follows, we focus on $\text{FNA}(x)$ in Section \ref{sec3} and  FNA to Section \ref{sec5}, respectively.   

\section{A Novel Sensitivity Analysis Method for $\textup{FNA}(x)$}  \label{sec3}
We propose a novel sensitivity analysis for $\text{FNA}(x)$, setting the Pearson correlation coefficient (PCC) $\rho(x) = \text{Corr}(Y^0, Y^1 \mid X=x)$ as the sensitivity parameter. 

\subsection{The Role of the PCC in Inferring $\textup{FNA}(x)$} 
\label{sec3-1}

 Before delving into the specific sensitivity analysis method, we elucidate why we choose $\rho(x)$  as the sensitivity parameter. Three key advantages arise from selecting $\rho(x)$ as the sensitivity parameter. 

First, $\rho(x)$, under Assumption \ref{assumption-1}, captures all information about the joint distribution  $\P(Y^0, Y^1 \mid X=x)$ for binary potential outcomes. Therefore, the range of $\rho(x)$ fully determines the range of $\text{FNA}(x)$, as shown in Proposition \ref{prop-1} below.   

\begin{proposition}
    \label{prop-1}  Under Assumption \ref{assumption-1},  the following statements are equivalent: 

(a)  the joint distribution $\P(Y^0, Y^1 \mid X=x)$ is identifiable;

(b) $\rho(x) = \textup{Corr}(Y^0, Y^1 \mid X=x)$ is identifiable; 

(c) $\textup{FNA}(x)$ is identifiable. 


\end{proposition} 

Proposition \ref{prop-1} indicates that 
 knowing $(\mu_0(x), \mu_1(x), \rho(x))$ is equivalent to knowing $\P(Y^0, Y^1 \mid X = x)$ under Assumption \ref{assumption-1}. 
The scalar  $\rho(x)$ is easy to interpret and offers a concise way to characterize the joint distribution.  
The equivalence between statements (a), (b) and (c) in Proposition \ref{prop-1} is intuitive because under Assumption \ref{assumption-1}, the joint distribution $\P(Y^0, Y^1 \mid X=x)$  involves four unknown parameters ($\pi_{jk}(x) = \P(Y^0 =j, Y^1 = k \mid X=x)$ for $j, k= 0, 1$) satisfying three equations, 
	$
	   \pi_{10}(x) + \pi_{11}(x)  = \mu_0(x),~  \pi_{01}(x) + \pi_{11}(x)  = \mu_1(x),~
	\sum_{j=0}^1 \sum_{k=0}^1  \pi_{jk}(x) = 1, 
	 $
and (b) or (c) contributes to another equation, thereby ensuring the identifiability of the joint distribution of the potential outcomes.

Second,  the $\rho(x)$ and $\text{FNA}(x)$ satisfy a linear relationship under Assumption \ref{assumption-1}. 
 Specifically,  by definition,  
   \[ \rho(x) = \frac{ \bfE[ Y^0 Y^1 \mid X =x ] - \bfE[ Y^0 \mid X=x ] \cdot  \bfE[ Y^1 \mid X=x ] } {\sqrt{ \V(Y^0 \mid X=x) \cdot  \V(Y^1 \mid X=x)}},     \]
 Under Assumption \ref{assumption-1}, the numerator equals $\P(Y^0=1, Y^1=1 \mid X) - \mu_0(x) \mu_1(x)  = \P( Y^0 = 1 \mid X=x) -  \P( Y^0 = 1, Y^1= 0 \mid X=x) - \mu_0(x) \mu_1(x) = \mu_0(x)(1-\mu_1(x)) - \text{FNA}(x)$ and is linear in $\text{FNA}(x)$,  and  
the denominator equals $\sqrt{  \mu_0(x)(1-\mu_0(x) ) \mu_1(x)(1-\mu_1(x) )}$, an identifiable quantity.  Therefore, $\rho(x)$  provides a straightforward description of $\text{FNA}(x)$.   

      

Last but not the least, adept and skilled practitioners are often capable of offering rough ranges for $\rho(x)$ or, at the very least, determining its sign in real-world applications, as illustrated in Section \ref{sec3-2} below with concrete examples. Thus, it is useful to investigate the sensitivity of $\text{FNA}(x)$ with respect to $\rho(x)$.  

There are other ways to delineate the joint distribution $\P(Y^0, Y^1 \mid X=x)$, such as the 
 risk difference, risk ratio and odds ratio between $Y^1$ and $Y^0$ given $X=x$. See Supplementary Material S1.1 for details.  
 Compared with these alternative measures, using $\rho(x)$ as the sensitivity parameter leads to a simpler form of bounds on $\mathrm{FNA}(x)$. 
\subsection{Improved Sharp Bounds on $\textup{FNA}(x)$} \label{sec3-2}
In this subsection, we demonstrate that the bounds in Lemma \ref{lem-1} could be tightened by invoking a mild assumption on $\rho(x)$. 
Before presenting the improved bounds, we first note that, given known marginal distributions of the potential outcomes,  the range values for $\rho(x)$ are not $[-1, 1]$ under Assumption \ref{assumption-1}.
\begin{proposition} \label{prop-range-rho}
     Under Assumption \ref{assumption-1}, we have that $\rho(x) \in [L_{\rho}(x), U_{\rho}(x)]$, where  
     		\begin{align*}
	  L_{\rho}(x) 
	  ={}& - \frac{ \min\{(1-\mu_0(x))(1-\mu_1(x)), \mu_0(x)\mu_1(x)\} }{\sqrt{ \mu_0(x)(1-\mu_0(x) ) \mu_1(x)(1-\mu_1(x) ) }} \leq 0, \\
	 U_{\rho}(x) ={}&  \frac{ \min\{\mu_0(x)(1-\mu_1(x)), \mu_1(x)(1-\mu_0(x)) \} }{\sqrt{ \mu_0(x)(1-\mu_0(x) ) \mu_1(x)(1-\mu_1(x) ) }} \geq 0.  
	\end{align*} 
These bounds are sharp. 	
\end{proposition}
Proposition \ref{prop-range-rho} gives the range of $\rho(x)$ under Assumption \ref{assumption-1}.
Next, we introduce Assumptions \ref{assumption-2}--\ref{assumption-3}  and show the sharp bounds on $\text{FNA}(x)$ under these assumptions.  

\begin{assumption}[Expert Knowledge] \label{assumption-2} $\rho_l(x) \leq  \rho(x)  \leq \rho_u(x)$, where $\rho_l(x)$ and $\rho_u(x)$ are pre-specified lower and upper bounds on $\rho(x)$. 
\end{assumption} 

When $\rho_l(x) = -1$ and $\rho_u(x) = 1$, Assumption \ref{assumption-2} holds trivially. 
 In Assumption \ref{assumption-2}, experienced practitioners may provide rough values of $\rho_l(x)$ and $\rho_u(x)$  at a given $x$. 
  However, in practice, we may not be sure about 
  the specific values of $\rho_l(x)$ and $\rho_u(x)$ due to lack of knowledge from experts. In such cases, we may instead consider the sign of $\rho(x)$, such as assuming $\rho(x) \geq 0$,  i.e., $Y^0$ and $Y^1$ are positively correlated given $X = x$.  
  The positive correlation assumption is generally mild in practice. 
   For example, in medical studies,  let $Y$ represent the disease status, a patient's health status (unmeasured factors)  affects both $Y^0$ and $Y^1$~\citep{efron1991compliance} and leads to a positive correlation between them.  
      Proposition \ref{prop-range-rho} implies that Assumption \ref{assumption-2}
      is uninformative when $\rho_l(x) < L_{\rho}(x)$ and $\rho_u(x) > U_{\rho}(x)$. To
       examine the impact of specifying an informative range for $\rho(x)$, we introduce Assumption \ref{assumption-3} below.  


\begin{assumption}[Information Gain Threshold] \label{assumption-3} $L_{\rho}(x)  \leq \rho_l(x)$ and $\rho_u(x) \leq U_{\rho}(x)$.
\end{assumption}  
Assumption \ref{assumption-3} provides informative thresholds for $(\rho_l(x), \rho_u(x))$ that
ensure Assumption \ref{assumption-2} meaningfully improves the bounds on $\mathrm{FNA}(x)$ over the Fr\'{e}chet--Hoeffding bounds. 
To simplify the presentation, let    
$$m(x) = \mu_0(x)(1-\mu_0(x) ) \mu_1(x)(1-\mu_1(x) ).$$   

\begin{theorem}[Proposed Sharp Bounds] \label{thm-1}  Assume Assumptions \ref{assumption-1} and \ref{assumption-2}. 

(a) $\mathrm{FNA}(x) \in [ L_{\mathrm{FNA}}(x),  U_{\mathrm{FNA}}(x)],$ 
    where 
		\begin{align*}
	   L_{\mathrm{FNA}}(x)  
					={}& \max\Big\{  \mu_0(x)(1 - \mu_1(x))  -  \rho_u(x) \sqrt{m(x)   },  \mu_0(x) - \mu_1(x),  0  \Big\}, \\ 
	 U_{\mathrm{FNA}}(x) 
					={}& \min \Big  \{ \max \big\{  \mu_0(x) (1 - \mu_1(x)) -    \rho_l(x) \sqrt{ m(x) },  0  \big\}, \mu_0(x),  1 - \mu_1(x)  \Big  \}.   
	\end{align*} 
These bounds on $\mathrm{FNA}(x)$ are sharp. 

(b) Further assume Assumption \ref{assumption-3}. The bounds in Theorem \ref{thm-1}(a) simplify to 
			\begin{equation}  \label{eq1} \begin{split}
	   L_{\mathrm{FNA}}(x)  
					={}&  \mu_0(x) (1 - \mu_1(x))  -  \rho_u(x) \sqrt{ m(x) }, \\ 
	 U_{\mathrm{FNA}}(x) 
					={}& \mu_0(x) (1 - \mu_1(x)) -    \rho_l(x) \sqrt{m(x) }.  
\end{split}	\end{equation} 
\end{theorem}

Theorem \ref{thm-1}(a) presents sharp bounds on $\mathrm{FNA}(x)$ under Assumptions \ref{assumption-1}--\ref{assumption-2}, which combine the bounds in \eqref{eq1}
  with the  Fr\'{e}chet--Hoeffding bounds. Therefore, the bounds in Theorem \ref{thm-1}(a) are at least as tight as the  Fr\'{e}chet--Hoeffding bounds in Lemma \ref{lem-1}.   
  When $(\rho_l(x), \rho_u(x))$ in Assumption \ref{assumption-2} are set as $\rho_l(x) =L_{\rho}(x)$ and $\rho_u(x) = U_{\rho}(x)$,  
 the bounds in Theorem \ref{thm-1}(a) reduce to the  Fr\'{e}chet--Hoeffding bounds. 
 Moreover, because $L_{\mathrm{FNA}}(x)$ (or $U_{\mathrm{FNA}}(x)$) is a monotonic function of $\rho_u(x)$ (or $\rho_l(x)$), 
 we could rewrite the bounds in Theorem \ref{thm-1}(a) as
  			\begin{equation} \label{eq2}  \begin{split}
	   L_{\mathrm{FNA}}(x)  
					={}&  \max \big  \{ \mu_0(x) (1 - \mu_1(x))  - \tilde  \rho_u(x) \sqrt{ m(x)  }, 0 \big \}, \\ 
	 U_{\mathrm{FNA}}(x) 
					={}&  \max \big  \{  \mu_0(x) (1 - \mu_1(x)) -   \tilde \rho_l(x) \sqrt{ m(x) }, 0 \big \}, 
\end{split}\end{equation}
where $\tilde{\rho}_l(x) = \max\{\rho_l(x), L_{\rho}(x)\}$ and $\tilde{\rho}_u(x) = \min\{\rho_u(x), U_{\rho}(x)\}$.

Theorem \ref{thm-1}(b) provides sharp bounds on $\mathrm{FNA}(x)$ under Assumptions \ref{assumption-1}--\ref{assumption-3},
where the lower and upper bounds are attained 
 at $\rho(x) = \rho_u(x)$ and $\rho(x) = \rho_l(x)$, respectively.  
By Proposition \ref{prop-1}, when $\rho_l(x) = \rho_u(x)$, the interval $[ L_{\mathrm{FNA}}(x),  U_{\mathrm{FNA}}(x)]$ collapses to a single point, indicating the identifiability of $\text{FNA}(x)$. This includes the identifiability results of \cite{Shen-etal2013} and \cite{Zhang-etal2013} under conditional independence of potential outcomes given covariates (i.e., $\rho_l(x) = \rho_u(x) = \rho(x) = 0$).      
 In particular, if $\rho_l(x) = 0$,  the upper bound simplifies to $U_{\text{FNA}}(x) = \mu_0(x)(1 - \mu_1(x))$. 

Theorem \ref{thm-1}(b) indicates that the width of the bounds, $U_{\text{FNA}}(x) -L_{\text{FNA}}(x)$, depends on $\rho_u(x) - \rho_l(x)$ and $\sqrt{ m(x)}$. The former is the range of $\rho(x)$. The latter equals $\sqrt{\V(Y^0 \mid X=x) \cdot \V(Y^1 \mid X=x)}$,  the product of the standard deviations of the two potential outcomes given $X=x$. This is intuitive as the variance $\V(Y^a \mid X=x)$ measures the variation in $Y^a$ that remains unexplained by $X$. It also implies that collecting more covariates correlated with outcome is beneficial for improving the bounds of FNA. 
\subsection{Sensitivity Analysis for $\text{FNA}(x)$}   \label{sec3-range}

Investigating the bounds under Assumptions \ref{assumption-1}--\ref{assumption-2} (or Assumptions \ref{assumption-1}--\ref{assumption-3}) establishes a sensitivity analysis framework for $\text{FNA}(x)$, with $\rho(x)$ being the sensitivity parameter. Different ranges of $\rho(x)$ correspond to different ranges of $\text{FNA}(x)$. See Supplementary Material S1.2 for an example for illustration.   



When the range of $\rho_l(x)$ and $\rho_u(x)$ is unrestricted a priori, one might want to report the possible value range of $\rho_u(x)$ that makes the lower bound in Theorem \ref{thm-1}(a) informative, i.e., the range of $\rho(x)$ that ensures $L_{\text{FNA}}(x)$ lies within the interval $(0, 1)$. 
We examine the threshold under Assumption \ref{assumption-2} and the case where $\tau(x) > 0$, 
since $L_{\text{FNA}}(x)$ is always positive when $\tau(x) < 0$.
Formally, we want to determine  $$\rho_{u}^*(x) := \min_{\rho_u(x)} \{ \rho_u(x) : L_{\mathrm{FNA}}(x) = 0 \},$$  
where $L_{\mathrm{FNA}}(x)$ denote the lower bound on $\mathrm{FNA}(x)$ in Theorem \ref{thm-1}(a). 

\begin{corollary}  \label{coro5} 
If $\tau(x) > 0$, then $$\rho_{u}^*(x) = \sqrt{ \frac{\mu_0(x) (1 - \mu_1(x) ) }{  \{(1-\mu_0(x))  \mu_1(x)\}}}, $$ 
which equals $U_{\rho}(x)$ in Proposition \ref{prop-range-rho} above.  
\end{corollary} 

Corollary \ref{coro5} shows that when $\tau(x) > 0$, $L_{\text{FNA}}(x) > 0$ if and only if $\rho(x) < \rho_u^*(x) = U_{\rho}(x)$.
Moreover, as $\tau(x)$ increases (via increasing $\mu_1(x)$ or decreasing $\mu_0(x)$), $\rho_u^*(x)$ decreases. In the extreme case where $\tau(x) = 1$, we have $\rho_u^*(x) = 0$. 
This suggests that when $\tau(x)$ is sufficiently large, a small value of $\rho_u(x)$ in Assumption \ref{assumption-2} suffices to ensure that $L_{\textup{FNA}}(x) = 0$.
Conversely, even when $\tau(x) > 0$, if it is not sufficiently large, the best-case scenario of assigning treatment to the subpopulation with $X = x$, namely $L_{\textup{FNA}}(x)$, may still result in harm. Hence, a positive $\tau(x)$ does not guarantee a small $\textup{FNA}(x)$.  
Consider a toy example with $\mu_1(x) = 1/2$ and $\mu_0(x) = 1/4$. Then $\tau(x) = 1/4 > 0$ and 
  $L_{\text{FNA}}(x) =   \max \{  (1-\sqrt{3}\rho_u(x))/8,  0\}.$
  Thus, $L_{\text{FNA}}(x) > 0$ provided that $\rho(x) \leq \rho_u^*(x) = 1/\sqrt{3} \approx 0.577$.  

Moreover, we can show that $\rho_{u}^*(x)$ can be expressed as a function of the odds ratio of $A$ on $Y$ conditional on $X=x$:  
 \begin{equation*} 
   \rho_{u}^*(x) =  \sqrt{1/\text{OR}_{AY}(x)},  
 \end{equation*}  
where   $$\text{OR}_{AY}(x) =\frac{\P(Y=1 \mid A=1, X=x) \P(Y =0 \mid A = 0, X=x) }{\P(Y =0 \mid A = 1, X=x) \P(Y = 1 \mid A = 0, X=x)}.  $$ 
 That is, the threshold $\rho_{u}^*(x)$ depends only on $\text{OR}_{AY}(x)$,  another important causal measure for binary outcomes under Assumption \ref{assumption-1}. The reformulation of $\rho_{u}^*(x)$ has two implications: 
 \begin{itemize}
 \item If $\text{OR}_{AY}(x)$ is only slightly larger than 1, then $L_{\text{FNA}}(x)$ is almost always larger than 0. For example, if $\text{OR}_{AY}(x) = 1.01$, then the threshold $\rho_{u}^*(x)$ is $0.995$, implying  that $L_{\text{FNA}}(x) > 0$ provided that $\rho_u(x) < 0.995$. 
 \item  If $\text{OR}_{AY}(x)$ is large, then it is hard to claim that $L_{\text{FNA}}(x)$ is larger than 0. For example, if $\text{OR}_{AY}(x) = 10000$, then $\rho_{u}^*(x)$ is $0.01$, i.e.,  $L_{\text{FNA}}(x) = 0$  if $\rho_u(x) \geq 0.01$. 
\end{itemize}

\subsection{Comparison and Illustration}  \label{sec4}  
In this section, we further compare the proposed bounds in Theorem \ref{thm-1} with the Fr\'{e}chet--Hoeffding bounds in Lemma \ref{lem-1},  and demonstrate the practical usage of the proposed sensitivity analysis method.
In the following, we examine the lower and upper bounds of $\text{FNA}(x)$ separately, corresponding to best- and worst-case scenarios. 

First, we compare the lower bound in Theorem \ref{thm-1}(b) with that in Lemma \ref{lem-1}(a).  In Lemma \ref{lem-1}(a), $L_{\text{FNA}}(x)$  is greater than 0 if and only if $\tau(x) < 0$. When $\tau(x) > 0$, the lower bound in  Lemma \ref{lem-1}(a) is always zero (no-harm) and provides no information. In contrast, the lower bound in  Theorem \ref{thm-1}(b) remains   informative even when $\tau(x) > 0$.  
The findings carry important practical implications for decision-making and policy evaluation. Clearly, in high-stakes application scenarios, opting for treatment should be avoided if its best-case scenario could result in substantial harm ($L_{\text{FNA}}(x) > 0$). 
 We give some illustrative examples below.

\begin{example}[{\bf Drug Authorization}] 
    Drug approval is typically granted based on the results of RCTs that confirm a significant positive ATE. 
    However, as revealed by Theorem \ref{thm-1} and Corollary \ref{coro5}, even with a significant positive ATE,  the drug may still pose harm to a considerable number of individuals. To address this concern, one possible approach is to estimate the lower bound on $\textup{FNA}$ and consider granting approval if the lower bound is sufficiently small or if the drug's effect significantly exceeds the lower bound~\citep{Bordley2009}. This strategy helps ensure that the best-case scenario is adequately safe or that the drug's benefits outweigh its potential harms, thereby enhancing the overall safety of approved medications.
  

As a real-world example in drug authorization, the Recombinant Human Activated Protein C Worldwide Evaluation of Severe Sepsis (PROWESS) clinical trial~\citep{Bernard-etal2001} randomized 1,690 patients with severe sepsis to receive the biologic drug Xigris. The intervention, an intravenously-administered recombinant human activated protein C,  exhibited a significant reduction in 28-day mortality by 6 percentage points compared with the control group with a 31\% mortality rate. Due to this substantial effect, the trial suspended enrollment, leading to the FDA's expedited approval of the drug for severe sepsis patients in 2001. 
However, the FDA approval of the drug was controversial due to an alternative analysis revealing a 1.5 percentage point increase in serious bleeding associated with the intervention~\citep{Siegel2002}. This increased risk of serious bleeding raised concerns about the potential harm to participants. Subsequently, the manufacturer voluntarily withdrew the drug from the market worldwide in 2011 based on further evidence~\citep{Food2011}.
\end{example}

\begin{example}[{\bf Policy Evaluation}]  
  In medical decision-making, treatment is typically assigned when $\tau(x) > c(x)$~\citep{2022nathan, Caron-etal2022, Eli-etal2022}, where $c(x)$ represents the cost of administering treatment. We denote the treatment policy as $d(x) = \mathbb{I}(\tau(x) > c(x))$.   
  As implied by Theorem \ref{thm-1} and Corollary \ref{coro5}, the best-case scenario for the treatment policy may result in harm. 
 One can use $L_{\textup{FNA}}(d) = \bfE[ L_{\textup{FNA}}(X) \cdot d(X) ]$ as a  criterion to further evaluate the policy. 
\end{example}

Next, we turn to explore the upper bound on $\text{FNA}(x)$.  
For the case of $\rho_l(x) \geq 0$ (positive correlation), 
the upper bound in Theorem \ref{thm-1}(b)  satisfies 
		\begin{align*}   
		U_{\mathrm{FNA}}(x)  ={}& \mu_0(x) \{ 1 - \mu_1(x) \} -    \rho_l(x)  \sqrt{  \mu_0(x)(1-\mu_0(x) ) \mu_1(x)(1-\mu_1(x) )}   \\
		\leq{}&   \mu_0(x) (1 - \mu_1(x)) \leq   \min\{ \mu_0(x), 1 - \mu_1(x)\},        
		 \end{align*}  
where the second inequality becomes an equality exclusively in four degenerate cases: $\mu_0(x)=0$ or 1, or $\mu_1(x) = 0$ or 1.   
Thereforce, positive correlation assumption significantly reduces the upper bound on $\text{FNA}(x)$. We provide further discussion on the upper bound in Supplementary Material S1.3.
 The upper bound on $\text{FNA}(x)$, representing the worst-case of FNA, is helpful in devising a treatment policy  with a controllable harm rate.   


\begin{example}[{\bf Policy Learning}]
Policy learning, a critical aspect of precision medicine, seeks to tailor treatment choices to individual patients' covariates. 
 Developing treatment policies with lower harm rate is important in high-stakes application scenarios~\citep{2022nathan, Eli-etal2022}.  
  This objective can be formulated as a constrained optimization problem: maximizing the average benefit while controlling the harm rate (or \textup{FNA}).  
   Specifically, for a given policy $d$, the harm rate and the expected reward induced by the policy are defined as  $\textup{FNA}(d) =  \bfE[ \textup{FNA}(X) \cdot d(X) ]$ and
       $\textup{R}(d) = \bfE[ \tau(X) \cdot d(X) ]$. Then, we could estimate the policy $d$  by solving the following optimization problem: 
             $\max_{d}   \textup{R}(d)~
s.t. ~\textup{FNA}(d) \leq \lambda,$     
where $\lambda$ is a pre-specified threshold.  However, \textup{FNA} is usually unidentifiable. 
In such cases, we could adopt different strategies to estimate a policy that balances the expected reward and harm rate~\citep{Arrow_Hurwicz_1977, Cui2021Individualized}: 
    \begin{itemize}
    \item \emph{Pessimistic strategy}. 
    Impose a constraint to ensure that the upper bound (worst-case) of \textup{harm rate} induced by a policy remains below $\lambda$, that is, 
           \begin{align*}  
             \max_{d}   \textup{R}(d) \quad 
s.t. \quad U_{\textup{FNA}}(d) =  \bfE[ U_{\textup{FNA}}(X) \cdot d(X) ]\leq \lambda,      
\end{align*} 
	\item \emph{Optimistic strategy}.  Impose  a constraint to ensure that the lower bound (best-case) of \textup{harm rate} induced by a policy remains below $\lambda$, that is,            
	       \begin{align*}  
             \max_{d}   \textup{R}(d) \quad 
s.t. \quad L_{\textup{FNA}}(d) =  \bfE[ L_{\textup{FNA}}(X) \cdot d(X) ]\leq \lambda,      
\end{align*} 
 
  	\item \emph{Mixed strategy}. Simultaneously take into account both the best- and worst-case harm rates, that is,  
   \begin{align*}  
             \max_{d}   \textup{R}(d) \quad 
s.t. \quad \alpha L_{\textup{FNA}}(d) + (1-\alpha) U_{\textup{FNA}}(d) \leq \lambda, 
\end{align*} 
where $\alpha$ is the ``coefficient of optimism", and $(1-\alpha)$  is the ``coefficient of pessimism". The value of $\alpha$ is pre-specified based on the problem at hand.	
    \end{itemize}    


\end{example} 

\section{Sensitivity Analysis for FNA} \label{sec5}

\subsection{Sensitivity Analysis}     
 In Section \ref{sec3}, we focus on $\text{FNA}(x)$, treating $\rho(x)$ as the sensitivity parameter.  When FNA is of interest, we can directly obtain its sharp bounds by taking the expectation of $\text{FNA}(X)$ over $X$. Specifically, under Assumptions \ref{assumption-1} and \ref{assumption-2}, the sharp lower and upper bounds on  FNA are given by
  \begin{align}\label{eq3} \begin{split}
 L_{\mathrm{FNA}}  ={}& \bfE \big [    \max \big  \{ \mu_0(X) (1 - \mu_1(X))  - \tilde  \rho_u(X) \sqrt{ m(X)  }, 0 \big \} \big ]  \\
 U_{\mathrm{FNA}}  ={}& \bfE  \big [   \max \big  \{  \mu_0(X) (1 - \mu_1(X)) -   \tilde \rho_l(X) \sqrt{ m(X) }, 0 \big \}  \big ], 
 \end{split}
 \end{align} 
where $\tilde{\rho}_l(x) = \max\{\rho_l(x), L_{\rho}(x)\}$ and $\tilde{\rho}_u(x) = \min\{\rho_u(x), U_{\rho}(x)\}$. 
However, applying the above bounds for $\mathrm{FNA}$ faces two practical challenges: (a) it requires specifying the range of $\rho(x)$ for each $x$; and (b) the forms of $L_{\mathrm{FNA}}$ and $U_{\mathrm{FNA}}$ are complex, posing significant difficulties for estimation and inference due to  the presence of multiple discontinuous $\min$ and $\max$ operators. 

Instead of using a function-valued sensitivity parameter $\rho(x)$, 
a simplified and more practical approach is to adopt a scalar sensitivity parameter $\rho$. 
Motivated by the structure of \eqref{eq3}, we propose the following form for a given interval $[\rho_l, \rho_u]$:
 \begin{equation}  \label{eq4}
  \beta_{\rho} 
	=  \bfE\big [ \max\big\{  \mu_0(X) ( 1 - \mu_1(X) )  -  \rho  \sqrt{  m(X)  }, 0 \big \}  \big ], \text{ for } \rho \in [\rho_l, \rho_u]      
	\end{equation}
to obtain  bounds on FNA and conduct sensitivity analysis accordingly. 
When $\tilde\rho_l(x) = \rho_l$ and $\tilde \rho_u(x) = \rho_u$ are constant in $x$,
 the bounds in \eqref{eq4} coincide with those in \eqref{eq3}.
 When $\tilde\rho_l(x)$ and $\tilde\rho_u(x)$ are not constant, we can interpret $\rho_l$ and $\rho_u$  in \eqref{eq4} as ``averaged" 
versions of $\tilde{\rho}_l(x)$ and $\tilde{\rho}_u(x)$ across individuals.  
 
 For \eqref{eq4}, an inappropriate specification of $\rho_l$ and $\rho_u$ may result in a bound on $\beta_{\rho}$  that is even worse than the Fr\'{e}chet--Hoeffding bounds on FNA. 
 However, this is not a major concern in real-world applications, as we can always compute the thresholds $\rho_l^*$ and $\rho_u^*$ such that $\beta_{\rho_l^*}$ and $\beta_{\rho_u^*}$ correspond to the Fr\'{e}chet--Hoeffding upper and lower bounds on FNA, respectively. These thresholds can serve as the range for $\rho$ in the absence of expert knowledge, ensuring that $\beta_{\rho}$ remains within the Fr\'{e}chet--Hoeffding bounds. 
  We focus on \eqref{eq4} in this paper and leave the general sensitivity analysis based on \eqref{eq3} to future research.

Next, we propose an estimator of $\beta_{\rho}$  based on the semiparametric efficiency theory~\citep{Bickel-etal1993, Tsiatis-2006}, and establish its large-sample properties.


\subsection{Estimation of the Bounds on FNA}  \label{sec6}
We express $\beta_{\rho}$ as 
	\begin{align*}  
	\beta_{\rho} ={}&  \bfE[ g(\bm{\eta}, \rho) \cdot \mathbb{I}\{g(\bm{\eta}, \rho) \geq 0 \}   ],  
	\end{align*}
where $g(\bm{\eta}, \rho) =  \mu_0(X) ( 1 - \mu_1(X) )   -  \rho  \sqrt{  \mu_0(X)(1-\mu_0(X) ) \mu_1(X)(1-\mu_1(X) )},$ 
with $\bm{\eta} = \bm{\eta}(X) :=  (e(X), \mu_0(X), \mu_1(X))$ being the nuisance parameters.   	 

When $\rho \leq 0$, $g(\bm{\eta}, \rho) \geq 0$, and the expression for $\beta_{\rho}$ simplifies to $\beta_{\rho} = \bfE[ g(\bm{\eta}, \rho)].$ In this case,  we can construct an efficient estimator for $\beta_{\rho}$ based on the semiparametric efficiency theory. 
However, when $\rho > 0$, $g(\bm{\eta}, \rho)$ may be negative, and the discontinuity of $\beta_{\rho}$ poses a major challenge for estimation and inference.   
Specifically, for $\rho > 0$,  $\beta_{\rho}$ does not possess an efficient influence function (EIF), as it is not pathwise differentiable due to its discontinuity~\citep{Bonvini-Kennedy2022, Kennedy2023-Semi}. In this scenario, if we directly employ plug-in estimators, the estimator for $g(\bm{\eta}, \rho)$  
may converge at most at the rate $1/\sqrt{n}$. Even if it converges at this rate, 
it will introduce significant estimation errors of order $1/\sqrt{n}$, which must be appropriately addressed and accounted for in inference. 
Unlike $\beta_{\rho}$ for $\rho > 0$, which may not be pathwise differentiable, certain terms in their expressions, $\beta_0 :=  \bfE[ \mu_0(X)(1-\mu_1(X) )]$ and 
$\gamma :=   \bfE[ \sqrt{  \mu_0(X)(1-\mu_0(X) ) \mu_1(X)(1-\mu_1(X) )}]$
are pathwise differentiable and have EIFs. Based on them, we can construct estimator  of $\beta_{\rho}$ for the case of $\rho > 0$. To proceed, we first derive the EIFs of $\beta_0$ and $\gamma$. 


\begin{theorem} \label{thm-3} Assume Assumption \ref{assumption-1}. 

(a)  The EIF of $\beta_0$ is $ \phi_\beta(Y,A,X; \bm{\eta})  - \beta_0$, where 
\begin{align*}   \phi_\beta(Y,A,X; \bm{\eta})  ={}&   \frac{(1-A)(Y - \mu_0(X))}{1-e(X)} (1 - \mu_1(X)) \\
 -{}& \frac{A (Y- \mu_1(X))}{e(X)} \mu_0(X) +  \mu_0(X)(1-\mu_1(X)). 
\end{align*}

(b) The EIF of $\gamma$  is $  \phi_\gamma(Y,A,X; \bm{\eta})  - \gamma$, where  
	\begin{align*}   \phi_\gamma(Y,A,X; \bm{\eta})  ={}&  \frac{ 1 - 2\mu_1(X) }{2} \sqrt{ \frac{ \mu_0(X) (1-\mu_0(X)) }{\mu_1(X)(1-\mu_1(X))} }   \frac{A(Y - \mu_1(X))}{e(X)} \\
  + {}&     \frac{ 1 - 2\mu_0(X) }{2} \sqrt{ \frac{ \mu_1(X) (1-\mu_1(X)) }{\mu_0(X)(1-\mu_0(X))} }   \frac{ (1-A)(Y - \mu_0(X))}{1 - e(X)} \\
  +{}&  
	 \sqrt{  \mu_0(X)(1-\mu_0(X) ) \mu_1(X)(1-\mu_1(X) )}.  
	\end{align*}

(c) In addition, the EIFs for $\beta_0$ and $\gamma$ remain the same no matter whether the propensity score $e(X)$ is known or not. 
\end{theorem}

Theorems \ref{thm-3}(a) and \ref{thm-3}(b) give the EIFs of $\beta_0$ and $\gamma$, respectively. Theorem \ref{thm-3}(c) states that the propensity score $e(X)$ is ancillary to the estimation of $\beta_0$ and $\gamma$, which is similar to the result on the estimation of ATE~\citep{Hahn1998}. In addition,  Theorem \ref{thm-3} implies that $\phi_\beta(Y,A,X; \bm{\eta})  -  \rho \cdot  \phi_\gamma(Y,A,X; \bm{\eta}) - \beta_{\rho}$ is the EIF of $\beta_{\rho}$ for $\rho\leq 0$.   
Next, we construct the estimators of $\beta_\rho$ (for both $\rho \leq 0$ and $\rho > 0$) based on Theorem \ref{thm-3}. Let  
 	\[  \varphi(Y, A, X; \bm{\eta}, \rho) =   \mathbb{I}\{g(\bm{\eta},\rho) \geq 0 \} \cdot \{ \phi_\beta(Y,A,X; \bm{\eta})  -  \rho \cdot  \phi_\gamma(Y,A,X; \bm{\eta}) \}. \]
Because $ \mathbb{I}\{ g(\bm{\eta}, \rho) \geq 0 \} $ is a function of $X$ and $\bfE[ \phi_\beta(Y,A,X; \bm{\eta})  -  \rho \cdot \phi_\gamma(Y,A,X; \bm{\eta}) \mid X] = g(\bm{\eta}, \rho)$, we can verify
 $\beta_{\rho} = \bfE[ \varphi(Y, A, X; \bm{\eta}, \rho)].$      
In line with \cite{Chernozhukov-etal-2018}, we adopt the cross-fitting technique to estimate the nuisance parameters $\bm{\eta}$; see Supplementary Material S1.4 for  details. Denote $\boldsymbol{\hat{\eta}} := (\hat e(X), \hat \mu_0(X), \hat \mu_1(X))$ as the cross-fitted estimators of $\bm{\eta}$, and define
\begin{equation} \label{est-b0}
\hat \beta_{\rho} = \frac{1}{n} \sum_{i=1}^n \varphi(Y_i, A_i, X_i; \boldsymbol{\hat{\eta}}, \rho)
\end{equation}
as an estimator of $\beta_{\rho}$.



  
  \subsection{Asymptotic Properties} \label{sec4-2}
We establish the large-sample properties for the proposed estimator $\hat \beta_{\rho}$  in (\ref{est-b0}). Let $||\cdot||_2$ represent  the $L_2$-norm, i.e., $|| \hat f(X) - f(X) ||_2 :=  \sqrt{\int (\hat f(x) - f(x))^2 d \mathbb{P}(x)}$ for generic functions $\hat f$ and $f$. 

\begin{condition} \label{con-1} Suppose that  
(a) $|| \hat e(X) - e(X)  ||_2 \cdot  || \hat \mu_a(X) - \mu_a(X) ||_2   = o_{\P}(n^{-1/2})$ for $a = 0, 1$;  
(b) $   || \hat \mu_1(X) - \mu_1(X) ||_2^2  +  || \hat \mu_0(X) - \mu_0(X) ||_2^2   = o_{\P}(n^{-1/2})$. 
\end{condition}

Condition \ref{con-1} is standard in machine-learning-aided causal inference~\citep{Belloni-etal-2017, Chernozhukov-etal-2018, Kennedy-2020}. Condition \ref{con-1} imposes weak restrictions on the estimation for the nuisance parameters and holds even when employing flexible machine learning methods. It only necessitates that the product of the $L_2$ errors in estimating $e(x)$ and $\mu_a(x)$ is of order $n^{-1/2}$, indicating that each regression function can be estimated at the slower rate of $n^{-1/4}$. 

Under Condition \ref{con-1}, we can obtain the asymptotic properties of $\hat \beta_\rho$ for $\rho \leq 0$. 

\begin{theorem} \label{thm-4} Under Assumption \ref{assumption-1} and Condition \ref{con-1},  
for $\rho \leq 0$, 
$\sqrt{n}(\hat \beta_\rho - \beta_\rho)     \xrightarrow{d} N(0,  \sigma_{\rho}^2),$ 
where $\sigma_{\rho}^2 = \mathbb{V}[ \varphi(Y, A, X; \bm{\eta}, \rho) ]$ is the efficiency bound of $\beta_\rho$. 
In addition,  a consistent estimator of $\sigma_{\rho}^2$ is 
$\hat \sigma_{\rho}^2 =   n^{-1} \sum_{i=1}^n [ \varphi(Y_i, A_i, X_i; \boldsymbol{\hat \eta}, \rho) - \hat \beta_\rho ]^2.$
  \end{theorem}

Theorem \ref{thm-4} shows that the estimator $\hat \beta_\rho$ for $\rho \leq 0$ is consistent, asymptotically normal, and locally efficient. 
However, for $\rho > 0$, the large-sample properties of $\hat \beta_{\rho}$ require additional Conditions \ref{con-2}--\ref{con-3} below. 

\begin{condition}[Margin condition] \label{con-2}
The random variable $g(\boldsymbol{\eta}, \rho)$ has absolutely continuous cumulative distribution function and there exist constants $C_0$ and $\alpha \geq 0$ such that  for all $t > 0$, $\P\big (  | g(\boldsymbol{ \eta}, \rho) |  \leq t  \big ) \leq C_0  t^{\alpha}.$ 
\end{condition}  
%
In Condition \ref{con-2}, the case where $\alpha = 0$ is trivial (no assumption) and included for notational convenience.  
If the density of $g(\boldsymbol{\eta}, \rho)$ is bounded by a constant $\delta$, then for any $t > 0$,
$\P\big( |g(\boldsymbol{\eta}, \rho)| \leq t \big) \leq 2 \delta t$, and 
Condition \ref{con-2} holds for $\alpha = 1$. 
Essentially, the parameter $\alpha$ determines how many cases are allowed to be close to the boundary, with a larger $\alpha$ meaning fewer cases are close~\citep{2018Who}.  
Condition \ref{con-2} characterizes the behavior of $g(\boldsymbol{\eta}, \rho)$ in the vicinity of the level $g(\boldsymbol{\eta}, \rho) = 0$, which is crucial for addressing the discontinuity of $\beta_p$ and accelerating the convergence rate of $\hat \beta_{\rho}$ for $\rho > 0$. The margin condition has been extensively studied and has proven useful in analyzing classification problem~\citep{Audibert-etal-2007, Laan-Luedtke2014} and other scenarios involving the estimation of nonsmooth functionals~\citep{Luedtke-etal-2016, Kennedy-2019}. 


\begin{condition} \label{con-3}  
$|| g(\boldsymbol{\hat \eta}, \rho) -  g(\boldsymbol{\eta}, \rho)  ||_\infty^{1 + \alpha}  = o_{\P}(n^{-1/2})$ for the $\alpha$ in Condition \ref{con-2}, where 
$|g(\boldsymbol{\hat \eta}, \rho) -  g(\boldsymbol{\eta}, \rho)  ||_\infty = \sup_{x\in \mathcal{X}} |g(\boldsymbol{\hat \eta}(x), \rho) -  g(\boldsymbol{\eta}(x), \rho)  |$. 
\end{condition}  


If the density of $g(\boldsymbol{\eta}, \rho)$ is bounded, achieving a convergence rate of $n^{-1/4}$ in the $L_{\infty}$-norm is adequate to fulfill Condition \ref{con-3}. This is because, in such a case, the margin condition holds for $\alpha = 1$.

\begin{theorem} \label{thm-5} Under Assumption \ref{assumption-1} and Conditions \ref{con-1}--\ref{con-3},  the estimator $\hat \beta_\rho$ for $\rho > 0$ satisfies that 
$\sqrt{n}(\hat \beta_{\rho} - \beta_{\rho})    \xrightarrow{d} N(0,  \sigma_{\rho}^2)$,  
where $\sigma_{\rho}^2 = \mathbb{V}[ \varphi(Y, A, X; \boldsymbol{\eta}, \rho) ]$.  In addition,  a consistent estimator of $\sigma_{\rho}^2$ is  $\hat \sigma_{\rho}^2 =   n^{-1} \sum_{i=1}^n [ \varphi(Y_i, A_i, X_i; \boldsymbol{\hat \eta}, \rho) - \hat \beta_{\rho} ]^2.$ 
  \end{theorem}

Theorem \ref{thm-5}  indicates that the proposed estimator $\hat \beta_{\rho}$ for $\rho > 0$ is $\sqrt{n}$-consistent and asymptotically normal, despite 
 the discontinuity of $\beta_{\rho}$. 
Although the asymptotic variance and variance estimator in Theorem \ref{thm-5} are the same as those based on the EIFs in Theorem \ref{thm-4}, the proof does not follow from the standard theory due to the non-smoothness of $\beta_\rho$ for $\rho > 0$.



Using a proof similar to that of Theorems \ref{thm-3}, the proposed estimation method can be readily extended to estimate $\text{FNA}(d)$, which represents FNA induced by a given policy $d$. Furthermore, by employing a similar proof as for Theorems \ref{thm-4} and \ref{thm-5}, we can derive the associated asymptotic properties for this extension. These results are omitted to avoid redundancy.

\section{Simulation} \label{sec7}
    


We use simulation to evaluate the finite-sample performance of the proposed estimator $\hat \beta_{\rho}$.  
Throughout the simulation, the covariates $X = (X_1, X_2, ..., X_p)^{\intercal} \sim N(0, I_p)$ with $I_p$ being an $p\times p$ identity matrix, 
the treatment $A$ follows a logistic regression  given by
    $\P(A=1\mid X) = \text{expit}((X_1 - X_2)/2),$
where $\text{expit}(x) = \exp(x)/\{1 + \exp(x) \}$. 
 The unmeasured variable $U \sim N(0, 1)$ 
 influences the potential outcomes $(Y^0,Y^1)$, thereby inducing correlation between them. The sample size $n$ is set at 500, 1000, and 2000, respectively.   
We consider three data-generating processes for $(Y^0, Y^1)$, 
\begin{enumerate} 
 \item[] {\bf (C1)} $\P(Y^0=1\mid X, U) = \text{expit}( (X_1 + X_2)/2 + 3U)$, and $\P(Y^1=1\mid X, U) = \text{expit}( (X_1 + X_2)/2 + 1 + 3U/2)$. 
 
   \item[] {\bf (C2)} $\P(Y^0=1\mid X, U) = \text{expit}( (X_1 + X_2)/2 + 2U)$, and $\P(Y^1=1\mid X, U) = \text{expit}( (X_1 + X_2)/2 + 1 + U)$. 

    \item[] {\bf (C3)} $\P(Y^0=1\mid X, U) = \text{expit}( (X_1 + X_2)/2 + U)$, and $\P(Y^1=1\mid X, U) = \text{expit}( (X_1 + X_2)/2 + 1 + U/2)$. 
\end{enumerate}

The only difference among cases (C1)--(C3) is the signal strength of $U$ relative to $X$. Thus, these three cases help to evaluate the impact of the signal-to-noise ratio on the proposed estimator $\hat \beta_{\rho}$. The true FNA values for cases (C1)--(C3) are 7.75\%, 9.34\%, and 11.27\%, computed by Monte Carlo. 
We also use Monte Carlo to compute the true value of $\beta_{\rho}$.  

We replicate each simulation 1,000 times and calculate the Bias, SD, ESE, and CP95 as evaluation metrics, where Bias and SD are the Monte Carlo bias and standard deviation of the points estimates, ESE and CP95 are the averages of estimated asymptotic standard error and coverage proportions of the 95\%  Wald-type confidence intervals based on Theorems \ref{thm-4} and \ref{thm-5}. 
For implementation, we employ logistic regression to estimate both the propensity score $e(X)$ and the outcome regression functions $\mu_a(X)$ for $a = 0, 1$, with two-fold cross-fitting~\citep{Chernozhukov-etal-2018}.   
Despite the mis-specification of $\mu_a(X)$ due to the $U$, we observe that logistic regression performs well (see Table \ref{tab1}).  

\begin{table*}[!t]
 \centering 
\caption{\centering Simulation results of $\hat \beta_{\rho}$ with $\rho = 0.0, 0.1, 0.2, 0.3, 0.4$ for cases (C1)-(C3).  \label{tab1}}
\scriptsize
\begin{tabular}{rrr | rrrr | rrrr | rrrr}
  \hline
      & &  &    \multicolumn{4}{c|}{$n =500$}    & \multicolumn{4}{c|}{$n =1000$} &  \multicolumn{4}{c}{$n =2000$}  \\
Case & $\rho$ & $\beta_{\rho}$ & Bias & SD & ESE & CP95 & Bias & SD & ESE & CP95 & Bias & SD & ESE & CP95 \\ 
  \hline
   & 0.0 & 0.160 & -0.000 & 0.020 & 0.021 & 0.954 & -0.000 & 0.014 & 0.014 & 0.949 & 0.000 & 0.010 & 0.010 & 0.942  \\ 
 & 0.1 &  0.137 & -0.001 & 0.020 & 0.020 & 0.948 & -0.000 & 0.014 & 0.014 & 0.946 & 0.000 & 0.010 & 0.010 & 0.948  \\ 
(C1) & 0.2 & 0.115 & -0.001 & 0.020 & 0.020 & 0.937 & -0.001 & 0.013 & 0.013 & 0.949 & -0.000 & 0.009 & 0.009 & 0.944 \\ 
 & 0.3 & 0.092 & -0.003 & 0.019 & 0.019 & 0.930 & -0.001 & 0.013 & 0.013 & 0.939 & -0.000 & 0.010 & 0.009 & 0.936  \\ 
 & 0.4 &  0.069 & -0.003 & 0.020 & 0.019 & 0.920 & -0.002 & 0.013 & 0.013 & 0.932 & -0.000 & 0.009 & 0.009 & 0.951 \\ 
   \hline  
    & 0.0 & 0.145 & -0.001 & 0.020 & 0.020 & 0.945 & 0.000 & 0.014 & 0.014 & 0.952 & 0.000 & 0.010 & 0.009 & 0.943 \\ 
   & 0.1 & 0.123 & -0.001 & 0.020 & 0.019 & 0.938 & 0.000 & 0.013 & 0.013 & 0.952 & -0.000 & 0.009 & 0.009 & 0.946 \\ 
(C2) & 0.2 & 0.102 & -0.001 & 0.019 & 0.019 & 0.945 & -0.001 & 0.013 & 0.013 & 0.956 & -0.000 & 0.009 & 0.009 & 0.959 \\ 
   & 0.3 & 0.080 & -0.002 & 0.019 & 0.018 & 0.930 & -0.001 & 0.013 & 0.012 & 0.947 & -0.000 & 0.008 & 0.009 & 0.950 \\ 
   & 0.4 & 0.058 & -0.005 & 0.019 & 0.017 & 0.906 & -0.001 & 0.012 & 0.012 & 0.923 & -0.001 & 0.008 & 0.008 & 0.945 \\ 
   \hline
    & 0.0 & 0.131 & -0.001 & 0.019 & 0.019 & 0.927 & -0.000 & 0.013 & 0.013 & 0.944 & 0.000 & 0.009 & 0.009 & 0.940  \\ 
   & 0.1 & 0.110 & -0.001 & 0.019 & 0.018 & 0.927 & -0.001 & 0.013 & 0.012 & 0.939 & -0.001 & 0.009 & 0.009 & 0.941 \\ 
  (C3) & 0.2 &  0.089 & -0.002 & 0.018 & 0.017 & 0.918 & -0.000 & 0.012 & 0.012 & 0.946 & -0.001 & 0.008 & 0.008 & 0.937  \\ 
   & 0.3 &  0.068 & -0.003 & 0.018 & 0.017 & 0.917 & -0.001 & 0.012 & 0.011 & 0.931 & -0.001 & 0.008 & 0.008 & 0.949  \\ 
   & 0.4 &  0.047 & -0.004 & 0.018 & 0.016 & 0.907 & -0.002 & 0.012 & 0.011 & 0.910 & -0.001 & 0.008 & 0.008 & 0.940  \\ 
   \hline 
\end{tabular}
\begin{flushleft}
Note: Bias and SD are the Monte Carlo bias and standard deviation over the 1,000 simulations of the points estimates of $\hat \beta_{\rho}$, ESE and CP95 are the averages of estimated asymptotic standard error and coverage proportions of the 95\%  Wald-type confidence intervals based on Theorems \ref{thm-4} and \ref{thm-5}.
\end{flushleft}
\end{table*}

Table \ref{tab1} presents the numerical results of estimator $\hat \beta_{\rho}$ for five representative points of 
$\rho = 0.0, 0.1, 0.2, 0.3, 0.4$, respectively. 
From Table \ref{tab1}, we have the following observations. First, Bias is small in all cases, indicating the consistency of the proposed estimator. Second,  as the sample size increases, ESE becomes closer to SD, and CP95 approaches its nominal value of 0.95,  demonstrating the validity of the asymptotic normality and variance estimation. Third, simulation results remain similar among cases (C1)--(C3), suggesting the stability of the proposed estimator in terms of signal-to-noise ratios of the potential outcomes.  


\begin{figure*}[!t]%
\centering
\includegraphics[scale = 0.75]{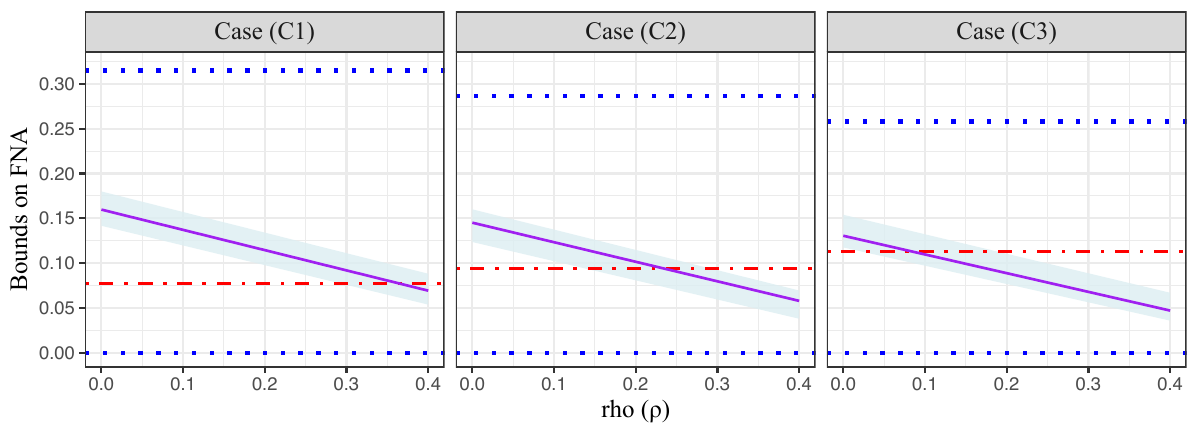}
\caption{Bounds on FNA for various values of $\rho$ in cases (C1)--(C3), based on a simulation with a sample size of 2,000. The dotdash red line represents the true value of FNA, the solid purple line is the true value of $\beta_{\rho}$, the dotted blue lines denote the Fr\'{e}chet--Hoeffding lower and upper bounds in Lemma \ref{lem-1}, and the shaded areas depict the 95\% Wald-type confidence intervals for $\hat \beta_{\rho}$.} \label{fig2} 
\end{figure*}


To showcase the effectiveness of our proposed sensitivity analysis method, we compare it with the Fr\'{e}chet--Hoeffding bounds in Lemma \ref{lem-1}. Fig. \ref{fig2} illustrates the bounds on FNA for different $\rho$ values in cases (C1)--(C3) using a simulation with a sample size of 2,000. The solid purple line represents the true $\beta_{\rho}$ value, the dotted blue lines indicate the Fr\'{e}chet--Hoeffding lower and upper bounds from Lemma \ref{lem-1}, and the shaded areas represent the 95\% confidence intervals for $\hat \beta_{\rho}$.  
 The findings drawn from Fig. \ref{fig2} are as follows: (a) 
 the shaded areas completely encompass the solid purple line, signifying that $\hat \beta_{\rho}$ accurately estimates $\beta_{\rho}$; (b) as expected, as $\rho$ increases, the true value of $\beta_{\rho}$ decreases, with a suitable $\rho$ where $\beta_{\rho}$ aligns with the true FNA value; (c) 
 notably, our proposed bounds significantly outperform the wide Fr\'{e}chet--Hoeffding bounds for $\rho \in [0, 0.4]$,  
 indicating a considerable improvement in FNA bounds due to the positive correlation assumption.

The dimension of the covariates in cases (C1)--(C3) is fixed at 2. However, in real-world data analysis, the dimension of the covariates is often higher, as seen in our application in Section \ref{sec8}, where the dimension of the covariates reaches 72. To further explore the impact of covariate dimension on the proposed estimator $\hat \beta_{\rho}$, we introduce three additional simulation cases, and the numerical results are similar to those in cases (C1)--(C3) and are presented in  Supplementary Material S1.5.



\section{Application to Right Heart Catheterization} \label{sec8}
Since the 1970s, right heart catheterization (RHC) has been a routine diagnostic and interventional procedure in hospitals, involving the insertion of a catheter into the pulmonary artery~\citep{Kubiak-etal2019}. Many critical care physicians believe that RHC helps improve patient outcomes. However, the advantages of RHC have not been proven through RCTs, as it would be unethical for physicians to engage in or encourage patients to participate in such a trial~\citep{Tan2006}. 
In this context, an influential yet controversial observational study by \cite{Connors-etal1996} raises concerns that RHC may not provide benefits to patients and could      
 cause harm to many  patients. Subsequently, several studies applied different statistical methods to assess the impact of RHC on survival by estimating the ATE~\citep{Hirano-Imbens2001, Shah-etal2005, Tan2006, Li-etal2008, Crump-etal2009, Vansteelandt-etal2012}.
In this section, we take it a step further and explore what percentage of patients are harmed by RHC using the same data from  \cite{Connors-etal1996}.  

\subsection{Data Description and Setup}
The study of \cite{Connors-etal1996} includes $n = 5,735$ critically ill patients admitted to the intensive care units (ICUs) of five medical centers between 1989 and 1994. For each patient, we observe a treatment status $A$ with $A=1$ indicating the use of RHC within the first 24 hours of  ICU care and $A = 0$ otherwise, an outcome $Y$ indicating whether survival $(Y=1)$ or not $(Y=0)$ up to 30 days, and 72 relevant covariates selected by medical experts, 
encompassing sociodemographic characteristics, primary disease category, secondary disease category, admission diagnosis categories, and comorbidities illness categories. \cite{Connors-etal1996} and \cite{Hirano-Imbens2001} offer a comprehensive description and summary statistics for these variables. 
Out of the 5,735 patients, 2,184 (38\%) received RHC and 3,551 (62\%) did not. In addition, among those who received RHC, about 38\%  died, compared with 31\% among those who did not receive RHC. Consequently, the death rate was around 7\% higher in patients receiving RHC than in those not receiving RHC. Nonetheless, this observation does not imply a causal link between the RHC and the increased likelihood of death, as the decision to recommend the procedure is not random. 
Doctors may be  more inclined to recommend RHC for patients they perceive to be in worse condition. 


Similar to the majority of previous studies, our method also relies on Assumption \ref{assumption-1}.   
To implement the proposed method, we need to estimate the propensity score $e(X)$ and the outcome regression functions $\mu_a(X)$ for $a = 0, 1$. We employ logistic regression to estimate them, using  the same variable selection method of \citet{Hirano-Imbens2001}. Both the propensity score and outcome regression functions are implemented with two-fold cross-fitting. Denote $\hat e(X)$, $\hat \mu_a(X)$ for $a = 0, 1$ as the estimated propensity score and outcome regression functions, respectively. Then,  the ATE of RHC on survival, using the doubly robust estimator~\citep{Bang-Robins-2005},  
 is $-0.055$ with an estimated standard error of 0.013. In other words, receiving RHC increases mortality by 5.5\%, aligning with the findings of most previous analyses.  
Additionally, with $\hat e(X)$ and $\hat \mu_a(X)$, we can estimate the Fr\'{e}chet--Hoeffding bounds on $\mathrm{FNA}$ in Lemma \ref{lem-1}, the bounds of $\rho(x)$ in Proposition \ref{prop-range-rho}, and the proposed estimator of $ \beta_{\rho}$. 

 \begin{figure}[t]%
\centering
\includegraphics[scale = 0.7]{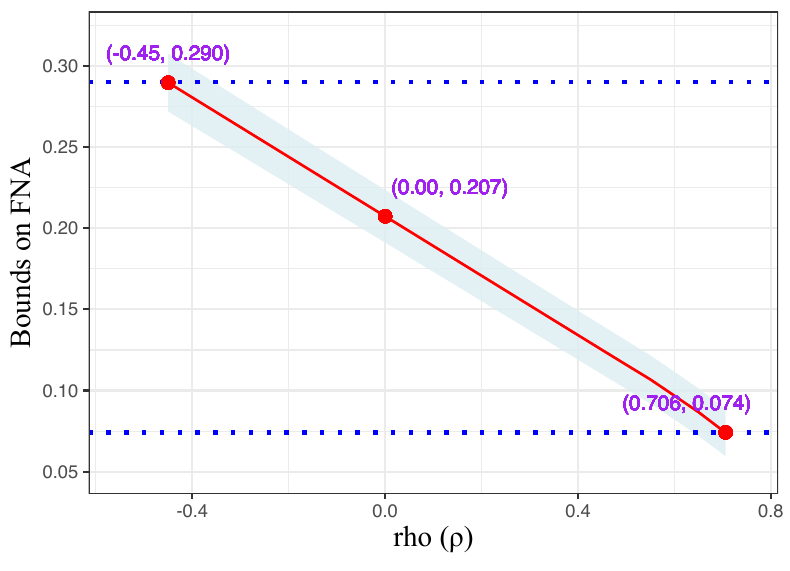}
\caption{The proposed estimator $\hat \beta_{\rho}$ for $\rho \in [-0.450, 0.706]$. The solid red line is the value of $\hat \beta_{\rho}$, the shaded areas depict the 95\% Wald-type confidence intervals for $\hat \beta_{\rho}$.  The dotted blue lines correspond to the estimated Fr\'{e}chet--Hoeffding lower and upper bounds in Lemma \ref{lem-1}. 
 Three points are highlighted in red, with their corresponding $\rho$-values (x-axis) and $\hat{\beta}_{\rho}$ values (y-axis) reported.} \label{app-fig2}
\end{figure}

\subsection{Results}


We present the proposed estimator $\hat \beta_{\rho}$ across varying values of $\rho$   
  in Fig. \ref{app-fig2}. This serves as a sensitivity analysis for FNA   with $\rho$ as the sensitivity parameter.
 From Fig. \ref{app-fig2}, we have the following observations: 
  (a) The estimator $\hat \beta_{\rho}$ is approximately linear in $\rho$ over the range $\rho \in [-0.450, 0.706]$, with a slope of about $-0.187$. This modest slope suggests that a one-unit increase in $\rho$ reduces FNA by only 18.7\%.  
 (b) The estimated Fr\'{e}chet--Hoeffding bounds on FNA fall within the interval [0.074, 0.290], with the endpoints corresponding to $\hat\beta_{\rho}$ at $\rho = -0.450$ and $\rho = 0.706$, respectively.  
Thus, if we identify a narrower range for $\rho$ than $[-0.450, 0.706]$, we can obtain tighter bounds than the  Fr\'{e}chet--Hoeffding bounds. (c) Considering that positive correlation is plausible in this study, since patients with good health status generally tend to have better outcomes, it is suitable to set $\rho_l = 0$.   
In this case, the proposed method yields an upper bound on FNA of 0.207,  significantly smaller than the Fr\'{e}chet--Hoeffding upper bound.    

\section{Discussion}  \label{sec9}
 In the main text, we focus on the bounds for $\text{FNA}(x)$, which is one of the four principal scores~\citep{Ding-Lu2017} defined by 
     $$\pi_{jk}(x) = \P(Y^0 =j, Y^1 = k \mid X=x) \text{ for } j, k= 0, 1. $$  
We derive sharp bounds for all $\pi_{jk}(x)$, with details in Supplementary Material S1.6. 
An important extension is to consider continuous outcomes. 
In such a case,  $\text{FNA}(x)$ is defined as $\text{FNA}(x) =   \P(Y^1 - Y^0 < 0 \mid X=x).$ 
 However, unlike the binary case,  $\rho(x)$ does not uniquely determine the joint distribution of general continuous potential outcomes, 
 making it difficult to establish a direct connection between $\rho(x)$ and $\mathrm{FNA}(x)$ without imposing additional assumptions.  
  To proceed, we can use copula to parametrize the joint distribution of the potential outcomes~\citep{Lu-etal2023, Chernozhukov-etal2023, Zhang-Yang2025}. 
  In the special case of bivariate normal potential outcomes:    
 \begin{equation*}  
 (Y^0, Y^1) \mid X=x\sim  \mathcal{N} \left(  \begin{pmatrix} \mu_0(x) \\   \mu_1(x)  \end{pmatrix}, 
      \begin{pmatrix} 
               \sigma_{0}^2(x) &   \rho(x) \sigma_{0}(x) \sigma_{1}(x)   \\ 
        \rho(x) \sigma_{0}(x) \sigma_{1}(x)  & \sigma_{1}^2(x) 
    \end{pmatrix} 
\right),   \end{equation*}
under Assumption \ref{assumption-1}, $(\mu_0(x), \mu_1(x), \sigma_0(x), \sigma_1(x))$ are identifiable, and the identifiability of joint distribution depends only on the PCC $\rho(x)$.  
 Under this normality assumption, given $\rho(x)$, we can obtain the closed-form of the identifiability formula for $\mathrm{FNA}(x)$, 
       $\mathrm{FNA}(x)  
          = \Phi \left ( - \tau(x) / \sigma_{\tau}(x)\right ),$  
  where $\sigma_{\tau}(x) = \sqrt{\sigma_1^2(x) + \sigma_0^2(x)  - 2 \rho(x) \sigma_0(x) \sigma_1(x)}$,  and $\Phi(\cdot)$ is the cumulative distribution function of a standard normal distribution. 
  Therefore, sharp bounds on $\mathrm{FNA}(x)$ can be readily obtained for a given range of $\rho(x)$.  
  For general continuous outcomes, let $\mathcal{C}$ denote the class of copulas that share a common range of $\rho(x)$:   
	$\mathcal{C} =  \{ C(\cdot, \cdot) \mid C(\cdot, \cdot) \text{ is a copula such that }  \rho(x) \in [\rho_l(x), \rho_u(x)] \}.$
In this case, it is challenging to derive the sharp bounds on $\text{FNA}(x)$, namely
$\inf_{C \in \mathcal{C}} \text{FNA}(x)$ and $\sup_{C \in \mathcal{C}} \text{FNA}(x)$. 
 Due to the technical complexity, we leave it to future work.

 When revising the manuscript, we noted that related to deriving sharp bounds on $\text{FNA}(x)$, \citet{Bodik-etal2025} focus on continuous outcomes and construct narrower prediction intervals for the individual treatment effect $Y^1 - Y^0$ by introducing a mild specification of $\rho(x)$ between the potential outcomes.
  We anticipate more future research on related topics.

\section{Competing interests}
No competing interest is declared.



\section{Acknowledgments}
The authors thank the anonymous reviewers for their valuable suggestions. 
Peng Wu was partially supported by the National Science Foundation of China (No. 12301370). Yue Liu was partially supported by the National Science Foundation of China (No. 12201629). Peng Ding was partially supported by the U.S. National Science Foundation (No. 1945136, No. 2514234). Peng Wu and Zhi Geng were partially supported by the BTBU Digital Business Platform Project by BMEC, the Beijing Key Laboratory of Applied Statistics and Digital Regulation, and the Academy for Interdisciplinary Studies at BTBU.


\bibliographystyle{plainnat}
\bibliography{reference}



%

 \newpage

  \begin{center}
\bf \Large 
Supplementary Material  for ``Quantifying  \\ Individual Risk for Binary Outcomes"
\end{center}

\setcounter{equation}{0}
\setcounter{section}{0}
\setcounter{figure}{0}
\setcounter{example}{0}
\setcounter{proposition}{0}
\setcounter{corollary}{0}
\setcounter{theorem}{0}
\setcounter{table}{0}
\setcounter{condition}{0}
\setcounter{lemma}{0}
\setcounter{remark}{0}

\renewcommand {\theproposition} {S\arabic{proposition}}
\renewcommand {\theexample} {S\arabic{example}}
\renewcommand {\thefigure} {S\arabic{figure}}
\renewcommand {\thetable} {S\arabic{table}}
\renewcommand {\theequation} {S\arabic{equation}}
\renewcommand {\thelemma} {S\arabic{lemma}}
\renewcommand {\thesection} {S\arabic{section}}
\renewcommand {\thetheorem} {S\arabic{theorem}}
\renewcommand {\thecorollary} {S\arabic{corollary}}
\renewcommand {\thecondition} {S\arabic{condition}}
\renewcommand {\thepage} {S\arabic{page}}
\renewcommand {\theremark} {S\arabic{remark}}

\setcounter{page}{1}

  \setcounter{equation}{0}
\renewcommand {\theequation} {S\arabic{equation}}
  \setcounter{lemma}{0}
\renewcommand {\thelemma} {S\arabic{lemma}}
   \setcounter{definition}{0}
\renewcommand {\thedefinition} {S\arabic{definition}}
   \setcounter{example}{0}
\renewcommand {\theexample} {S\arabic{example}}
   \setcounter{proposition}{0}
\renewcommand {\theproposition} {S\arabic{proposition}}
   \setcounter{corollary}{0}
\renewcommand {\thecorollary} {S\arabic{corollary}}

 \bigskip 
 

 \section{Additional Results}

  \subsection{Alternative Measures}
  
 The following discussion corresponds to the remark at the end of Section \ref{sec3-1}.
 
  There are other ways to delineate the joint distribution $\P(Y^0, Y^1 \mid X=x)$. For example, we can define the risk difference (RD), risk ratio (RR) and odds ratio (OR) between $Y^1$ and $Y^0$ given $X=x$ as 
 \begin{align*}
 \text{RD}(x) ={}& \P(Y^1=1 \mid  Y^0=1, X=x) - \P(Y^1=1 \mid  Y^0=0, X=x), \\
 \text{RR}(x) ={}&\dfrac{\P(Y^1=1 \mid  Y^0=1, X=x)}{\P(Y^1=1 \mid  Y^0=0, X=x)}, \\
 \text{OR}(x) ={}& \frac{\P(Y^1=1 \mid  Y^0=1, X=x) \P(Y^1 =0 \mid  Y^0 = 0, X=x)}{\P(Y^1 =0 \mid  Y^0 = 1, X=x) \P(Y^1 = 1 \mid  Y^0 = 0, X=x)}.
 \end{align*}  

The following Proposition \ref{prop-2} establishes the connection between $\rho(x)$ and these three measures. 

\begin{proposition} \label{prop-2}  Assume Assumption \ref{assumption-1}.

(a) The following five statements are equivalent: $Y^0 \indep Y^1 \mid X=x$,  $\rho(x)=0$, $\textup{RD}(x)=0$, $\textup{RR}(x)=1$, and $\textup{OR}(x)=1$. 

(b) If $\P(Y^0=a, Y^1=b \mid X=x) > 0$ for all $a,b=0,1$, then the following four statement are  equivalent:  $\rho(x) > 0$, $\textup{RD}(x)>0$, $\textup{RR}(x)>1$, and $\textup{OR}(x)>1$. 

\end{proposition}

Compared with the three alternative measures in Proposition \ref{prop-2}, using $\rho(x)$ as the sensitivity parameter leads to a simpler form of bounds on $\mathrm{FNA}(x)$.

\begin{proof}[Proof of Proposition \ref{prop-2}.]      

  By the definitions of $\text{RD}(x)$, $\text{RR}(x)$, and $\text{OR}(x)$, we can verify  the equivalence of $\text{RD}(x)=0$, $\text{RR}(x)=1$, and $\text{OR}(x)=1$, as well as the equivalence of $\text{RD}(x)>0$, $\text{RR}(x)>1$, and $\text{OR}(x)>1$.  
This establishes the relationship among  $\text{RD}(x)$, $\text{RR}(x)$, and $\text{OR}(x)$. 

To prove Proposition \ref{prop-2}, it therefore suffices to establish the relationship between $\rho(x)$ and $\text{RD}(x)$. We claim that 
\begin{equation} \label{eq-s1} \begin{split}
 \rho(x)  \sqrt{ \V(Y^0 \mid X=x) \cdot  \V(Y^1\mid X=x)   }  = \text{RD}(x) \cdot \P(Y^0=1\mid X=x)\cdot \P(Y^0=0\mid X=x).
\end{split}
\end{equation}  
If  (\ref{eq-s1}) holds, then $\rho(x) = 0$ implies $\text{RD}(x) = 0$, establishing Proposition \ref{prop-2}(a). Moreover, if $\rho(x) > 0$, then $\text{RD}(x) > 0$, which establishes Proposition \ref{prop-2}(b). 

Next, we show that  (\ref{eq-s1}) holds. We have 
\begin{align*}
\text{RD}&(x) 
={} \P(Y^1=1 \mid Y^0=1, X=x)-   \P(Y^1 = 1  \mid Y^0=0, X=x) \\
={}&  \frac{ \P(Y^1=1, Y^0=1  \mid X=x)  }{ \P(Y^0=1  \mid X=x) } - \frac{\P(Y^1 = 1,  Y^0=0  \mid X=x)}{\P(Y^0=0  \mid X=x)} \\  
=&  \frac{ \P(Y^1=1, Y^0=1  \mid X=x)\P(Y^0=0  \mid X=x)}{\P(Y^0=1  \mid X=x)\P(Y^0=0  \mid X=x) }  - \frac{  \P(Y^1 = 1,  Y^0=0  \mid X=x) \P(Y^0=1 \mid X=x) }{\P(Y^0=1 \mid X=x)\P(Y^0=0  \mid X=x) } \\
=&  \frac{ \P(Y^1=1, Y^0=1  \mid X=x)\P(Y^0=0  \mid X=x) -  \P(Y^1 = 1,  Y^0=0  \mid X=x) \P(Y^0=1 \mid X=x) }{\P(Y^0=1  \mid X=x)\P(Y^0=0  \mid X=x) }     
\end{align*}
and therefore, 
\begin{align*}
&\rho(x) \cdot \sqrt{ \V(Y^0  \mid X=x) \cdot  \V(Y^1 \mid X=x) } \\
={}& \bfE(Y^1Y^0  \mid  X=x)-\bfE(Y^1  \mid  X=x) \cdot \bfE(Y^0 \mid X=x) \\
={}& \P(Y^1=1, Y^0=1  \mid X=x)  - \P(Y^1=1  \mid X=x) \cdot \P(Y^0=1  \mid X=x)   \\
={}& \P(Y^1=1, Y^0=1  \mid X=x)  \\ 
{}& -  \P(Y^0=1  \mid X=x)  \times \{ \P(Y^1=1, Y^0 = 0  \mid X=x) + \P(Y^1=1, Y^0 = 1  \mid X=x) \}  \\
={}& \P(Y^1=1, Y^0=1  \mid X=x)\cdot \P(Y^0=0  \mid X=x)  - \P(Y^1=1, Y^0 = 0  \mid X=x) \cdot \P(Y^0=1  \mid X=x) \\
={}& \text{RD}(x) \cdot  \P(Y^0=1  \mid X=x)\P(Y^0=0  \mid X=x).  
\end{align*}
This completes the proof. 

\end{proof}

 \subsection{An Example for Sensitivity Analysis}
 The following example corresponds to the remark at the beginning of Section \ref{sec3-range}.
 
 \begin{example} Suppose that the data-generating process for the potential outcomes $(Y^0, Y^1)$ is  
 	\begin{align*} 
		\P(Y^0=1  \mid U) ={}& \frac{\exp(U)}{1+\exp(U)}, 	\\
	\P(Y^1=1 \mid U) ={}& \frac{\exp(1+U)}{1+\exp(1+U)},  
	\end{align*}
 where $U$ is an unmeasured variable taking value in $\{0, 2\}$ uniformly. Based on Monte Carlo, we obtain the true FNA value is 8.85\%, the PCC between $Y^0$ and $Y^1$ is 0.128, $\mu_1 = 0.842$, and $\mu_0 = 0.690$. For $\rho \in [0, 1]$ (positive correlation),  the bounds in Theorem \ref{thm-1} are $$\max\{ \mu_0(1-\mu_1)- \rho \sqrt{\mu_0(1-\mu_0)\mu_1(1-\mu_1)}, 0 \} = \max\{0.109-0.169\rho, 0\}, \quad \rho \in [0, 1].$$ The upper and lower bounds in Lemma \ref{lem-1} are 0.158 and 0, respectively. Figure \ref{fig1} displays the results.  As shown in Figure \ref{fig1},  the upper bound of the proposed method (at $\rho = 0$) is significantly smaller than that in Lemma \ref{lem-1}, and the lower bound is larger than that in Lemma \ref{lem-1} when $\rho$ is not sufficiently large.  

 \smallskip 
\begin{figure}[H]%
\centering
\includegraphics[scale = 0.80]{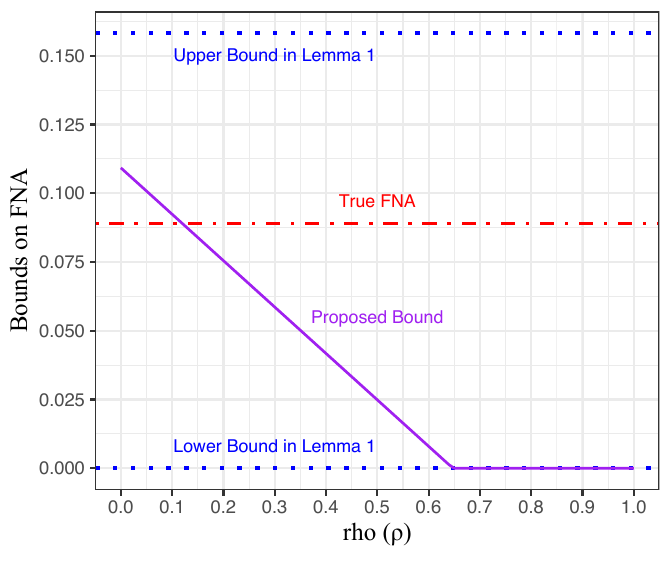}
\caption{\centering Bounds on FNA for various values of $\rho$.}
 \label{fig1}
\end{figure}    
 \end{example}

 \subsection{Further Analysis of the Upper Bounds on FNA}
 The following discussion corresponds to the remark in Section \ref{sec4} on the upper bounds on FNA in Lemma \ref{lem-1} and Theorem \ref{thm-1}.  

\begin{proposition}[Upper Bounds]  \label{prop-3}  
\, 

(a)  The upper bound $U_{\text{FNA}}(x)$ in Lemma \ref{lem-1}(a) satisfies  
	\[  U_{\text{FNA}}(x) =   \min\{  \mu_0(x), 1 - \mu_1(x) \}  \leq \frac{1-\tau(x)}{2}.  \]

(b) When $\rho_l(x) \geq 0$,  the upper bound $U_{\text{FNA}}(x) $ in Theorem \ref{thm-1}  satisfies that   
 \[ U_{\text{FNA}}(x) \leq  \mu_0(x) \{ 1 - \mu_1(x) \}   \leq (1-\tau(x))^2/4. \] 



	
\end{proposition}

Proposition  \ref{prop-3} suggests that as $\tau(x)$ increases from 0 to 1, the upper bound gradually decreases from 1/2 by \ref{prop-3}(a) (or 1/4 by \ref{prop-3}(b))   to 0.   



\begin{proof}[Proof of Proposition \ref{prop-3}.]  
  Observe that $1 - \tau(x) = \mu_0(x) + (1 - \mu_1(x))$, 
and for ease of presentation, let 
\[    a = \mu_0(x), \quad b = 1 - \mu_1(x).     \]
The inequalty in  Proposition \ref{prop-3}(a) follows from 
\[       \min\{a, b\} \leq \frac{a+b}{2} \]  
for positive constants $a$ and $b$.   

 
The inequality in  Proposition \ref{prop-3}(b) follows from 
 $$U_{\text{FNA}}(x) \leq  \mu_0(x) \{ 1 - \mu_1(x) \}  = a b  \leq  \left( \frac{ a +  b }{2} \right )^2 =   (1-\tau(x))^2/4.$$ 
 This completes the proof. 
 
\end{proof}

\subsection{Estimation Based on Cross-Fitting} 
The following discussion corresponds to the remark in Section \ref{sec6}.  
 The proposed estimation procedure of 
   \[ \beta_{\rho} 
	=  \bfE\Big [ \max\Big\{  \mu_0(X) ( 1 - \mu_1(X) )  -  \rho  \sqrt{  \mu_0(X)(1-\mu_0(X) ) \mu_1(X)(1-\mu_1(X) )    }, 0 \Big \}  \Big ]     \]  
  is summarized in Algorithm 1. 


\noindent\rule[0.25\baselineskip]{\textwidth}{1pt}
{\bf Algorithm 1.}    We partition the data into $K$ distinct groups, each containing $n/K$ observations (for simplicity, assuming equal-sized groups, but not mandatory), and define $I_{1}$ through $I_{K}$ as the corresponding index sets. Additionally, let $I_{k}^{C} = \{1, ...,  n\}\setminus I_k$ denote the complement of $I_{k}$ for $k = 1, \ldots, K$.

\noindent 
{\bf Step 1 (nuisance parameter training with cross-fitting).}  

 \medskip
  {\bf for } $k = 1$ {\bf to} $K$ {\bf do}
 		
(1) Construct estimates $\boldsymbol{\hat \eta}_{-k} = (\hat e(x), \hat \mu_1(x), \hat \mu_0(x))$   
using the subsample indexed by $I_{k}^{C}$. 

(2) Obtain the predicted values of $\boldsymbol{\eta}(X_i)$ for $i \in I_k$, denoted by  $\boldsymbol{\hat \eta}_{-k}(X_i)$.

{\bf end}	
%

 All the predicted values of $\boldsymbol{\hat \eta}_{-k}(X_i)$  for $i \in \{1,2,..., n\}$ consist of the final estimates of $\boldsymbol{\eta}(X_i)$, denoted by $\boldsymbol{\hat \eta}(X_i).$

\medskip  \noindent 
{\bf Step 2 (estimation of $\beta_{\rho}$).} The proposed estimator of $\beta_\rho$ is  
\begin{equation*}  
\hat \beta_{\rho} =  \frac{1}{n}  \sum_{i=1}^n   \varphi(Y_i, A_i,X_i; \boldsymbol{\hat  \eta}, \rho). 
\end{equation*}  
\noindent\rule[0.25\baselineskip]{\textwidth}{1pt}

\bigskip 
In Algorithm 1, the full sample is divided into $K$ subsets. For each subset indexed by $I_k$, the nuisance parameters are trained in the corresponding complementary subset $I_k^C$ while predictions are made within $I_k$. 
The predicted values of nuisance parameters from all subsets collectively form the final estimates.  Subsequently, the proposed estimator of $\beta_{\rho}$ is obtained as a sample average of  $\varphi(Y, A,X; \boldsymbol{\hat  \eta}, \rho).$
The cross-fitting approach, as outlined in Algorithm 1, has been widely adopted recently~\citep{wager2018estimation, Athey-Tibsirani-Wager-2019, Kennedy-2020}. In Step 1 of Algorithm 1, we allow for the use of flexible machine learning techniques to estimate the nuisance parameters.

  \subsection{Additional Results for Simulation}  \label{app-F}
The following discussion corresponds to the remark in Section \ref{sec7}.   

To further explore the impact of the covariate dimension on the proposed estimator $\hat \beta_{\rho}$, we introduce three additional simulation cases:   
 \begin{enumerate} 
 \item[] {\bf (C4)} $\P(Y^1=1\mid X, U) = \text{expit}( X^{\intercal}\alpha + 1 + U)$ and $\P(Y^0=1\mid X, U) = \text{expit}( X^{\intercal}\alpha + U)$, where $\alpha = (1,1/2, ..., 1/2^{p-1})^{\intercal}$,  $p = 20$ is the dimension of $X$.  
 
   \item[] {\bf (C5)} The data-generation mechanism is the same as in case (C4), except that $p = 50$. 

    \item[] {\bf (C6)}  The data-generation mechanism is the same as in case (C4), except that $p = 100$.  
\end{enumerate}
In cases (C4)--(C6), we set the covariate dimension at  $p = 20, 50, 100$, respectively. All covariates are relevant with nonzero coefficients on potential outcomes, but only a few are truly important with large coefficients. This represents a common approximate sparsity scenario. Instead of using logistic regression directly, we adopt the logistic regression with $L_1$ penalty to mitigate overfitting, 
and the corresponding tuning parameters are determined by using 5-fold cross-validation.


\begin{table*}[!ht]
\centering 
\caption{\centering Simulation results of $\hat \beta_{\rho}$ with $\rho = 0.0, 0.1, 0.2, 0.3, 0.4$ for cases (C4)--(C6).  \label{tab-s2}}
\scriptsize
\begin{tabular}{rrr | rrrr | rrrr | rrrr}
\hline
& &  &    \multicolumn{4}{c|}{$n =500$}    & \multicolumn{4}{c|}{$n =1000$} &  \multicolumn{4}{c}{$n =2000$}  \\
Case & $\rho$ & $\beta_{\rho}$ & Bias & SD & ESE & CP95 & Bias & SD & ESE & CP95 & Bias & SD & ESE & CP95 \\ 
\hline
& 0.0 & 0.128 & -0.000 & 0.023 & 0.022 & 0.944 & 0.000 & 0.015 & 0.015 & 0.950 & -0.001 & 0.010 & 0.010 & 0.956 \\ 
\multirow{2}{*}{(C4)} & 0.1 & 0.108 & -0.004 & 0.023 & 0.022 & 0.938 & -0.002 & 0.015 & 0.015 & 0.936 & -0.003 & 0.010 & 0.010 & 0.930 \\ 
& 0.2 & 0.089 & -0.004 & 0.023 & 0.022 & 0.932 & -0.003 & 0.015 & 0.015 & 0.936 & -0.005 & 0.009 & 0.010 & 0.938 \\ 
& 0.3 & 0.069 & -0.006 & 0.023 & 0.022 & 0.920 & -0.004 & 0.015 & 0.015 & 0.936 & -0.004 & 0.010 & 0.010 & 0.930 \\ 
\hline  
& 0.0 & 0.128 & -0.003 & 0.029 & 0.024 & 0.938 & 0.002 & 0.015 & 0.015 & 0.944 & -0.000 & 0.011 & 0.010 & 0.956 \\ 
\multirow{2}{*}{(C5)}  & 0.1 & 0.108 & -0.005 & 0.029 & 0.025 & 0.952 & -0.000 & 0.016 & 0.015 & 0.934 & -0.002 & 0.010 & 0.010 & 0.944 \\ 
& 0.2 & 0.089 & -0.007 & 0.029 & 0.024 & 0.928 & -0.002 & 0.015 & 0.015 & 0.954 & -0.003 & 0.010 & 0.010 & 0.946 \\ 
& 0.3 & 0.069 & -0.008 & 0.029 & 0.025 & 0.932 & -0.004 & 0.015 & 0.015 & 0.938 & -0.004 & 0.010 & 0.010 & 0.942 \\ 
\hline
& 0.0 & 0.128 & -0.027 & 0.093 & 0.041 & 0.890 & 0.002 & 0.016 & 0.015 & 0.944 & 0.001 & 0.011 & 0.010 & 0.952 \\ 
\multirow{2}{*}{(C6)} & 0.1 & 0.108 & -0.032 & 0.108 & 0.042 & 0.914 & 0.001 & 0.015 & 0.015 & 0.954 & -0.001 & 0.010 & 0.010 & 0.954 \\ 
& 0.2 & 0.089 & -0.025 & 0.093 & 0.037 & 0.912 & -0.000 & 0.015 & 0.015 & 0.952 & -0.002 & 0.010 & 0.010 & 0.948 \\ 
& 0.3 & 0.069 & -0.044 & 0.125 & 0.043 & 0.894 & -0.004 & 0.015 & 0.015 & 0.952 & -0.003 & 0.010 & 0.010 & 0.940 \\ 
\hline 
\end{tabular}
\begin{flushleft}
Note: Bias and SD are the Monte Carlo bias and standard deviation over the 1,000 simulations of the points estimates of $\hat \beta_{\rho}$, ESE and CP95 are the averages of estimated asymptotic standard error and coverage proportions of the 95\%  Wald-type confidence intervals based on Theorems \ref{thm-4} and \ref{thm-5}.
\end{flushleft}
\end{table*}

\begin{figure*}[!t]%
\centering
\includegraphics[scale = 0.8]{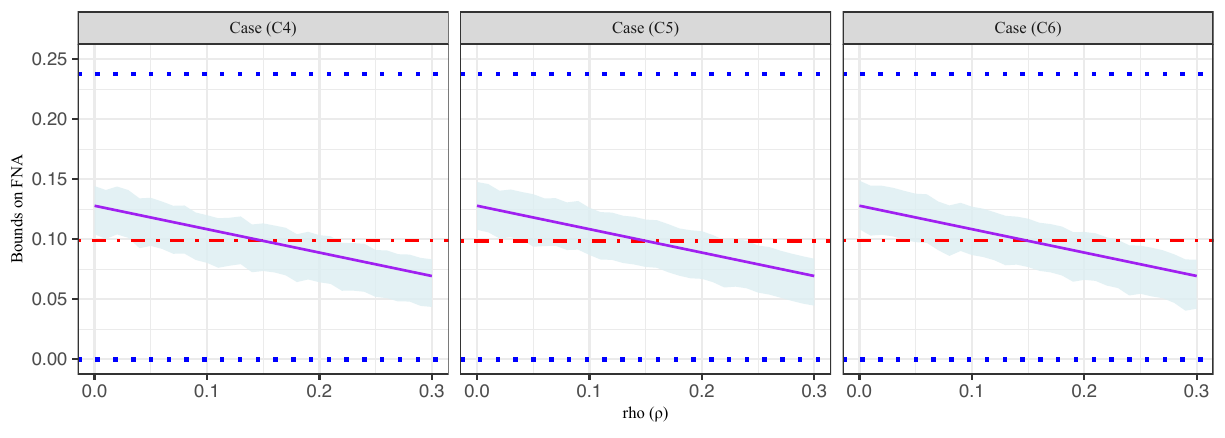}
\caption{Bounds on FNA for various values of $\rho$ in cases (C4)-(C6), based on a simulation with a sample size of 2,000. The dotdash red line represents the true value of FNA, the solid purple line is the true value of $\beta_{\rho}$, the dotted blue lines denote the Fr\'{e}chet--Hoeffding lower and upper bounds in Lemma \ref{lem-1}, and the shaded areas depict the 95\% confidence intervals for $\hat \beta_{\rho}$.} \label{fig-s1}
\end{figure*}


The simulation results for cases (C4)--(C6) are similar to those in cases (C1)--(C3) of the manuscript,  and are presented in Table \ref{tab-s2} and Figure \ref{fig-s1}, which indicates that our method performs well even when the dimension of covariates is moderately high.


\subsection{Extension to Other Estimands}

The following discussion corresponds to the remark in Section \ref{sec9}.  In the main text, we focus on analyzing the bounds on $\text{FNA}(x)$, which is one of the four principal scores~\citep{Ding-Lu2017}, defined as 
     \[ \pi_{jk}(x) = \P(Y^0 =j, Y^1 = k \mid X=x) \text{ for } j, k= 0, 1. \]
Recall $m(x) = \mu_0(x)(1-\mu_0(x) ) \mu_1(x)(1-\mu_1(x) )$. The following Proposition \ref{prop-10} (The proofs are provided at the end of the next section)  shows the sharp bounds on the remaining three principal scores under Assumptions 1--2 and Assumptions 1--3, respectively. 

\begin{proposition} \label{prop-10}  We have that

(a)   Under Assumptions 1--2, 
\begin{itemize}
\item  $\pi_{11}(x) \in [ L_{\pi_{11}}(x),  U_{\pi_{11}}(x)],$
where 
		\begin{align*}
	   L_{\pi_{11}}(x)
					={}& \max \Big \{  \mu_0(x) \mu_1(x) + \rho_l(x) \sqrt{ m(x)}, \mu_0(x) + \mu_1(x) - 1, 0 \Big \}, \\ 
	 U_{\pi_{11}}(x) 
					={}& \min \Big\{  \max \big \{\mu_0(x)  \mu_1(x) +  \rho_u(x) \sqrt{m(x)}, 0 \big \}, \mu_0(x), \mu_1(x) \Big \}.  
	\end{align*} 
These bounds on $\pi_{11}(x)$ are sharp.

\item $\pi_{01}(x) \in [ L_{\pi_{01}}(x),  U_{\pi_{01}}(x)],$
where 
\begin{align*}
L_{\pi_{01}}(x)
			={}& \max \Big \{  (1-\mu_0(x)) \mu_1(x)    -   \rho_u(x) \sqrt{ m(x) },  \mu_1(x) - \mu_0(x), 0 \Big\}, \\ 
U_{\pi_{01}}(x) 
					={}&  \min\Big\{ \max \big \{ (1-\mu_0(x)) \mu_1(x) -   \rho_l(x) \sqrt{m(x)}, 0 \big \}, 1-\mu_0(x), \mu_1(x)\Big\}. 
	\end{align*} 
 These bounds on $\pi_{01}(x)$ are sharp.

\item $\pi_{00}(x) \in [ L_{\pi_{00}}(x),  U_{\pi_{00}}(x)],$
where 
		\begin{align*}
	   L_{\pi_{00}}(x)
					={}& \max\Big\{ (1-\mu_0(x))(1- \mu_1(x)) + \rho_l(x) \sqrt{ m(x)}, 1-\mu_0(x)-\mu_1(x), 0 \Big\}, \\ 
	 U_{\pi_{00}}(x) 
			={}&   \min\Big\{ \max \big \{ (1-\mu_0(x))(1- \mu_1(x)) + \rho_u(x) \sqrt{m(x)}, 0 \big \},  1- \mu_0(x), 1- \mu_1(x) \Big \}.  
	\end{align*} 
These bounds on $\pi_{00}(x)$ are sharp. 

\end{itemize}

(b)  Under Assumptions 1--3, the bounds in Proposition S2(a) simplify to 

\begin{itemize}

\item  $\pi_{11}(x) \in [ L_{\pi_{11}}(x),  U_{\pi_{11}}(x)],$
where 
		\begin{align*}
	   L_{\pi_{11}}(x)
					={}&  \mu_0(x) \mu_1(x) + \rho_l(x) \sqrt{ m(x)}, \\ 
	 U_{\pi_{11}}(x) 
					={}& \mu_0(x)  \mu_1(x) +  \rho_u(x) \sqrt{m(x)}.  
	\end{align*} 
These bounds on $\pi_{11}(x)$ are sharp.

\item $\pi_{01}(x) \in [ L_{\pi_{01}}(x),  U_{\pi_{01}}(x)],$
where 
\begin{align*}
L_{\pi_{01}}(x)
			={}&  (1-\mu_0(x)) \mu_1(x)    -   \rho_u(x) \sqrt{ m(x) }, \\ 
U_{\pi_{01}}(x) 
					={}& (1-\mu_0(x)) \mu_1(x) -   \rho_l(x) \sqrt{m(x)}. 
	\end{align*} 
 These bounds on $\pi_{01}(x)$ are sharp.

\item $\pi_{00}(x) \in [ L_{\pi_{00}}(x),  U_{\pi_{00}}(x)],$
where 
		\begin{align*}
	   L_{\pi_{00}}(x)
					={}& (1-\mu_0(x))(1- \mu_1(x)) + \rho_l(x) \sqrt{ m(x)  }, \\ 
	 U_{\pi_{00}}(x) 
					={}&  (1-\mu_0(x))(1- \mu_1(x)) + \rho_u(x) \sqrt{m(x)}.  
	\end{align*} 
These bounds on $\pi_{00}(x)$ are sharp.

\end{itemize}
\end{proposition}     

The bounds on $\pi_{jk}(x)$ in Proposition \ref{prop-10} have a form similar to those for $\text{FNA}(x)$. In addition, improving the bounds on $\text{FNA}(x)$ directly leads to tighter bounds on the other three principal scores. Specifically, under the ignorability assumption, the joint distribution $\P(Y^0, Y^1 \mid X=x)$  involves four unknown parameters ($\pi_{jk}(x)$ for $j, k= 0, 1$) satisfying three equations: 
	\begin{align} \label{R1}
	   \pi_{10}(x) + \pi_{11}(x)  = \mu_0(x),~  \pi_{01}(x) + \pi_{11}(x)  = \mu_1(x),~
	\sum_{j=0}^1 \sum_{k=0}^1  \pi_{jk}(x) = 1, 
	 \end{align}
where $\text{FNA}(x) = \pi_{10}(x)$, $\mu_a(x)=\mathbb{E}[Y\mid X=x, A=a]$ for $a = 0, 1$. 
Equation \eqref{R1} contains only one free parameter. By taking $\pi_{10}(x)$ as the free parameter, the remaining three parameters can be expressed in terms of it: 
	\begin{align*} \begin{cases}
		\pi_{11}(x) ={} \mu_0(x) -  \text{FNA}(x), \\
		\pi_{01}(x) ={}  \mu_1(x) - \mu_0(x) +  \text{FNA}(x), \\
		\pi_{00}(x) ={}   1 - \mu_1(x) - \text{FNA}(x). 
	\end{cases}\end{align*}
	Therefore, given the marginal distributions of potential outcomes, since $\text{FNA}(x)$ determines all other principal scores, the improved bounds on $\text{FNA}(x)$ can directly lead to tighter bounds on the other three principal scores.

These bounds in Proposition \ref{prop-10} broaden the applicability of the proposed framework. For example, the results of Proposition \ref{prop-10} can be applied to analyze probability of causation. We provide an example below.

\bigskip 
\begin{example}[Attribution Analysis]  
Causal inference involves not only evaluating the effects of causes but also  the causes of  effects~\citep{Dawid2022, pearl2009causality, Pearl-etal2016-primer, Pearl-Mackenzie2018}.   
The probability of sufficiency (PS) and the probability of necessity (PN) are the routinely used quantities for attribution defined as  
   \begin{align*}
   \textup{PS}(A \Rightarrow Y)    ={}& \P( Y^1= 1 \mid A = 0, Y = 0 ), \\
   \textup{PN}(A \Rightarrow Y)  ={}& \P(  Y^0 = 0 \mid A=1, Y=1 ). 
   \end{align*}
Take $\textup{PN}(A \Rightarrow Y)$ as an example. It can be written as 
   \[     \frac{ \P(  Y^0 = 0, Y^1 = 1 \mid A=1 )}{ \P(Y^1=1 \mid A=1)},   \]
where the denominator is an identifiable quantity and the numerator equals to $\bfE[ \pi_{01}(X) \mid A=1] = \bfE[ e(X) \pi_{01}(X) ] /\P(A=1)$ under Assumption \ref{assumption-1}. Thus, we could derive the bounds on $\textup{PN}(A \Rightarrow Y)$ by applying Proposition \ref{prop-10}(b).  
\end{example}

\section{Proofs} 

\subsection{Proof of Lemma \ref{lem-1}} \label{app-B}

\begin{proof}[Proof of Lemma \ref{lem-1}.] By the  
Fr\'{e}chet--Hoeffding inequality, 
\begin{align}
&  \P(Y^0=1, Y^1 = 0 \mid X=x)  \notag \\
\leq{}& \min \big \{  \P(Y^0=1 \mid X=x),    \P(Y^1 = 0\mid X=x) \big \} \label{eq-tmp1}   \\
={}&   \min \big\{   \mu_0(x),   1 - \mu_1(x) \big\}   \notag
\end{align} 
and 		
\begin{align}  
& \P(Y^0=1, Y^1 = 0 \mid X=x) \notag \\
\geq{}&  \max \big \{  \P(Y^0=1 \mid X=x) + \P(Y^1 = 0 \mid X=x) - 1, 0 \big \} \label{eq-tmp2} \\
={}&   \max \big \{ \mu_0(x) - \mu_1(x) , 0 \big\}.  \notag
\end{align} 
Next, we show these bounds are sharp, i.e., they are achievable.  
Note that $$\P(Y^0=1, Y^1 = 1 \mid X=x) = 0 \iff \mu_0(x) = \P(Y^0 = 1 \mid X =x) = \P(Y^0 = 1, Y^1 = 0 \mid X =x),$$ and  $$\P(Y^0=0, Y^1 = 0 \mid X=x) = 0 \iff  1 - \mu_1(x) = \P(Y^1 = 0 \mid X =x) = \P(Y^0 = 1, Y^1 = 0 \mid X =x).$$ 
Thus,  \eqref{eq-tmp1} becomes equality if and only if $\P(Y^0=1, Y^1 = 1 \mid X=x) = 0$ or $\P(Y^0=0, Y^1 = 0 \mid X=x) = 0$. 

Also, since $\mu_0(x) - \mu_1(x) =  \P(Y^0=1 \mid X=x) - \P(Y^1 = 1 \mid X=x) =  \P(Y^0=1, Y^1 = 0 \mid X=x) - \P(Y^0 = 0, Y^1 = 1 \mid X=x)$, then \eqref{eq-tmp2} becomes equality if and only if $\P(Y^0=0, Y^1 = 1 \mid X=x) = 0$ or $\P(Y^0=1, Y^1 = 0 \mid X=x) = 0$.  

\end{proof}


  
  \subsection{Proof of Proposition \ref{prop-1}}
\begin{proof}[Proof of Proposition \ref{prop-1}.]    Under Assumption \ref{assumption-1}, the joint distribution $\P(Y^0, Y^1 \mid X=x)$ involves four unknown parameters 
$$\pi_{jk}(x) = \P(Y^0 =j, Y^1 = k \mid X=x) \text{ for } j, k= 0, 1.$$ They satisfy three equations:  
 \begin{align*} 
 \begin{cases}
 & \pi_{10}(x) + \pi_{11}(x)  = \mu_0(x), \\
 &  \pi_{01}(x) + \pi_{11}(x)  = \mu_1(x), \\
  & \sum_{j=0}^1 \sum_{k=0}^1  \pi_{jk}(x) = 1. 
 \end{cases}
 \end{align*} 
  The equivalence between (a) and (c) follows immediately by noting that $\text{FNA}(x) = \pi_{10}(x)$.

Next, we show the equivalence between (b) and (c). 
  Note that 
\[ \rho(x) = \frac{ \bfE[ Y^0 Y^1 \mid X =x ] - \bfE[ Y^0 \mid X=x ] \cdot  \bfE[ Y^1 \mid X=x ] } {\sqrt{ \V(Y^0 \mid X=x) \cdot  \V(Y^1 \mid X=x)   } }    \]
where under Assumption \ref{assumption-1},  the numerator is equal to 
	 \begin{align*}
	 & \P(Y^0=1, Y^1=1 \mid X=x) - \mu_0(x) \mu_1(x) \\
	  ={}& \P( Y^0 = 1 \mid X=x) -  \P( Y^0 = 1, Y^1= 0 \mid X=x) - \mu_0(x) \mu_1(x)  \\
	 ={}&  \mu_0(x)(1-\mu_1(x)) - \text{FNA}(x), 
       \end{align*}
  and the denominator is an identifiable quantity that equals 
  		\[  \sqrt{ \mu_0(x)(1 - \mu_0(x))  \mu_1(x)(1 - \mu_1(x))   }. \]   
	Therefore,  the identifiability of $\rho(x)$ is  equivalent to the   identifiability of $ \text{FNA}(x)$ under Assumption \ref{assumption-1}.  This finishes the proof.  
		
\end{proof}

\subsection{Proof of Proposition \ref{prop-range-rho}} 
\begin{proof}[Proof of Proposition \ref{prop-range-rho}.]
By the definition of $\rho(x)$, under Assumption \ref{assumption-1}, we have 
\begin{align*} 
\rho(x) ={}& \frac{ \bfE[ Y^0 Y^1 \mid X =x ] - \bfE[ Y^0 \mid X=x ] \cdot  \bfE[ Y^1 \mid X=x ] } {\sqrt{ \V(Y^0 \mid X=x) \cdot  \V(Y^1 \mid  X=x) } } \\
={}&    \frac{\P(Y^0=1, Y^1=1 \mid X=x) - \mu_0(x)\mu_1(x) }{\sqrt{  \mu_0(x)(1-\mu_0(x) ) \mu_1(x)(1-\mu_1(x) ) }}.
\end{align*} 
Applying the Fr\'{e}chet--Hoeffding inequality to $\P(Y^0=1, Y^1=1 \mid X=x)$ yields  
\begin{align*}
\P(Y^0=1, Y^1=1 \mid X=x) \geq{}& \max\{\mu_0(x) + \mu_1(x) -1, 0 \}, \\
\P(Y^0=1, Y^1=1 \mid X=x) \leq{}& \min\{\mu_0(x),  \mu_1(x)\},
\end{align*}  
which imply  
       		\begin{align*}
	\rho(x) \geq  L_{\rho}(x) 
	  = - \frac{ \min\{(1-\mu_0(x))(1-\mu_1(x)), \mu_0(x)\mu_1(x)\} }{\sqrt{  \mu_0(x)(1-\mu_0(x) ) \mu_1(x)(1-\mu_1(x) ) }}	\end{align*}
and 
\[  \rho(x) \leq 
	 U_{\rho}(x) =  \frac{ \min\{\mu_0(x)(1-\mu_1(x)), \mu_1(x)(1-\mu_0(x)) \} }{\sqrt{  \mu_0(x)(1-\mu_0(x) ) \mu_1(x)(1-\mu_1(x))  }}.  
 \]	   
 This finishes the proof. 
  
\end{proof}

\subsection{Proof of Theorem \ref{thm-1}}

\begin{proof}[Proof of Theorem \ref{thm-1}.] We prove Theorem \ref{thm-1}(a) and Theorem \ref{thm-1}(b), separately. 

\medskip \noindent 
\emph{Proof of Theorem \ref{thm-1}(a)}. 
{\bf First,} under Assumption \ref{assumption-2},  
  $\text{Corr}(Y^0, Y^1 \mid X=x) \geq \rho_l(x)$, and therefore, 
\begin{align*}
&  \bfE\Big [ \{Y^0 - \bfE(Y^0 \mid X=x)\} \cdot \{Y^1 - \bfE(Y^1 \mid X=x)\}  \mid  X=x \Big ]  \\ 
\geq{}& \rho_l(x) \cdot \sqrt{ \V(Y^0\mid X=x) \cdot  \V(Y^1\mid X=x)   }   
\end{align*}     
which is equivalent to 
\begin{align*}
{}&   \bfE[ Y^0 Y^1 \mid X =x ] - \bfE[ Y^0 \mid X=x ] \cdot  \bfE[ Y^1 \mid X=x ] \\
\geq{}&  \rho_l(x) \cdot \sqrt{ \V(Y^0\mid X=x) \cdot  \V(Y^1\mid X=x)   }    \\ 
\iff{}&   \P( Y^0 = 1, Y^1  = 1\mid X=x)  - \P(Y^0 = 1\mid X=x ) \cdot   \P(Y^1 = 1\mid X=x )  \\
\geq{}&  \rho_l(x) \cdot \sqrt{ \V(Y^0\mid X=x) \cdot  \V(Y^1\mid X=x)   }   \\
\iff{}&   \P( Y^0 = 1\mid X=x) -  \P( Y^0 = 1, Y^1= 0 \mid X=x)  - \P(Y^0 = 1\mid X=x ) \cdot   \P(Y^1 = 1\mid X=x )    \\
\geq{}&  \rho_l(x) \cdot \sqrt{ \V(Y^0\mid X=x) \cdot  \V(Y^1\mid X=x)   }    \\
\iff{}&  \text{FNA}(x) \leq{}  \P( Y^0 = 1\mid X=x)  - \P(Y^0 = 1\mid X=x ) \cdot   \P(Y^1 = 1\mid X=x )  \\
{}& - \rho_l(x) \cdot \sqrt{ \V(Y^0\mid X=x) \cdot  \V(Y^1\mid X=x)   }   \\
={}& \mu_0(x)\{1 - \mu_1(x) \}  - \rho_l(x) \sqrt{  \mu_0(x)(1-\mu_0(x) ) \mu_1(x)(1-\mu_1(x) )}.
\end{align*}
The inequality becomes equality  if and only if  $\text{Corr}(Y^0, Y^1 \mid X=x) = \rho_l(x)$.  

Similarly, $\text{Corr}(Y^0, Y^1 \mid X=x) \leq \rho_u(x)$ implies that 
\begin{align*}
&   \bfE[ \{Y^0 - \bfE(Y^0 \mid X=x)\} \{Y^1 - \bfE(Y^1 \mid X=x)\}  \mid  X=x ]   \\
\leq{}&  \rho_u(x) \cdot \sqrt{ \V(Y^0\mid X=x) \cdot  \V(Y^1\mid X=x)   },   
\end{align*}     
which is equivalent to 
\begin{align*}
& \P( Y^0 = 1\mid X=x) -  \P( Y^0 = 1, Y^1= 0 \mid X=x)   - \P(Y^0 = 1\mid X=x ) \cdot   \P(Y^1 = 1\mid X=x )  \\
\leq{}& \rho_u(x) \cdot   \sqrt{ \V(Y^0\mid X=x) \cdot  \V(Y^1\mid X=x) }, 
\end{align*}
that is, 
\begin{align*}
&\text{FNA}(x) \geq{}    \P( Y^0 = 1\mid X=x) \cdot \P(Y^1 = 0 \mid X= x)   -  \rho_u(x) \cdot  \sqrt{ \V(Y^0\mid X=x) \cdot  \V(Y^1\mid X=x) } \\
={}&  \mu_0(x)\{1 - \mu_1(x) \} - \rho_u(x) \sqrt{  \mu_0(x)(1-\mu_0(x) ) \mu_1(x)(1-\mu_1(x) )}.
\end{align*}	 
The inequality becomes equality  if and only if  $\text{Corr}(Y^0, Y^1 \mid X=x) = \rho_u(x)$.

{\bf Second,} since $\text{FNA}(x)$ is a probability, we must have  $0 \leq L_{\text{FNA}}(x) \leq U_{\text{FNA}}(x) \leq 1$, which implies 
    \[      \text{FNA}(x)  \leq \min \Big \{  \max \Big \{ \mu_0(x)\{1 - \mu_1(x) \} - \rho_l(x) \sqrt{  \mu_0(x)(1-\mu_0(x) ) \mu_1(x)(1-\mu_1(x) )}, 0 \Big  \},  1 \Big  \}   \]
  and 
       \[      \text{FNA}(x)  \geq \min \Big \{ \max \Big \{ \mu_0(x)\{1 - \mu_1(x) \} - \rho_u(x) \sqrt{  \mu_0(x)(1-\mu_0(x) ) \mu_1(x)(1-\mu_1(x) )}, 0 \Big  \}, 1 \Big \}.   \]
We find the above operators $\min\{\cdot, 1 \}$ are unnecessary, as 
	\begin{align*}
\mu_0(x)\{1 - \mu_1(x) \}  +  \sqrt{  \mu_0(x)(1-\mu_0(x) ) \mu_1(x)(1-\mu_1(x) )} \leq{}& 1
	\end{align*}
always holds. This is because for $a= \mu_0(x), b = 1 - \mu_1(x)$, $0\leq a, b \leq 1$,  
		\begin{align*}
		 &   ab +  \sqrt{  a b (1-a)(1-b) }  
	 \leq{} 1, \\
	 \iff {}&   a b (1-a)(1-b)  \leq (1-ab)^2 \\
	 \iff{}&   a^2 b + ab^2 - 3 a b + 1 \geq 0. 
		\end{align*}
Observe that  $a^2 b + ab^2 - 3 a b + 1 = b\{a^2  - ( 3 - b )a \} + 1$. For  any fixed $b \in [0, 1]$,  treat this expression as a quadratic function of $a$ over the interval $a \in [0, 1]$. Since its vertex occurs at $( 3  - b)/2 \geq 1$, the minimum over $a\in [0, 1]$ is attained at the boundary point  $a = 1$. In this case, the minimum is $b^2 - 2 b + 1$, which is nonnegative. Therefore, $ a^2 b + ab^2 - 3 a b + 1 \geq 0$ for all $a \in [0, 1]$ and $b \in [0, 1]$.  

{\bf Third}, by Lemma \ref{lem-1},  we note that under Assumption \ref{assumption-1},  
                      \[    \text{FNA}(x)  \geq \max\{ \mu_0(x) - \mu_1(x), 0 \}      \]
and 
         \[       \text{FNA}(x)  \leq \min\{ \mu_0(x),  1- \mu_1(x) \}.     \]

 Theorem \ref{thm-1}(a) follows by combining the above bounds on $\text{FNA}(x)$.

         \medskip \noindent 
\emph{Proof of Theorem \ref{thm-1}(b)}. 
         To prove Theorem \ref{thm-1}(b), it suffices to show that under additional Assumption \ref{assumption-3},  
        \begin{align*}
         \begin{cases}
          \mu_0(x)\{1 - \mu_1(x) \} - \rho_l(x) \sqrt{  \mu_0(x)(1-\mu_0(x) ) \mu_1(x)(1-\mu_1(x) )} \geq 0,   \\
               \mu_0(x)\{1 - \mu_1(x) \} - \rho_u(x) \sqrt{  \mu_0(x)(1-\mu_0(x) ) \mu_1(x)(1-\mu_1(x) )} \geq 0, \\
          \mu_0(x)\{1 - \mu_1(x) \} - \rho_l(x) \sqrt{  \mu_0(x)(1-\mu_0(x) ) \mu_1(x)(1-\mu_1(x) )} \leq          \min\{ \mu_0(x),  1- \mu_1(x) \},  \\
          \mu_0(x)\{1 - \mu_1(x) \} - \rho_u(x)  \sqrt{  \mu_0(x)(1-\mu_0(x) ) \mu_1(x)(1-\mu_1(x) )}  \geq   \max\{ \mu_0(x) - \mu_1(x), 0 \}. 
         \end{cases}
        \end{align*}
          The above four inequalities follow immediately by the proof of Proposition \ref{prop-range-rho}.

           This completes the proof. 
%


\end{proof}

\subsection{Proof Corollary \ref{coro5}}

\begin{proof}[Proof of Corollary \ref{coro5}.]  
 We want to determine 
	\begin{equation*}  
	\rho_{u}^*(x) = \min_{\rho_u(x)} \{ \rho_u(x):  L_{\text{FNA}}(x) = 0 \}, 
	\end{equation*} 
where $L_{\text{FNA}}(x)$ is the lower bound of $\text{FNA}(x)$ in Theorem 3(a). For a given $x$,  $L_{\text{FNA}}(x)$ is a monotone function of $\rho_u(x)$, and therefore, $\rho_{u}^*(x)$ has a unique solution: 
\[     \rho_{u}^*(x) =  \sqrt{  \frac{   \mu_0(x) (1 - \mu_1(x) )  }{  (1-\mu_0(x))  \mu_1(x)   }  }, \] 
which equals to $U_{\rho}(x)$ in Proposition  \ref{prop-range-rho} when $\tau(x) > 0$. 

\end{proof}

%


  \subsection{Proof of Theorem \ref{thm-3}}

\begin{proof}[Proof of Theorem \ref{thm-3}.]  We obtain the EIFs of $\beta_0$ and $\gamma$ separately. 

\medskip 
{\bf EIF of $\beta_0$}. 
Let $f(x)$ be the density function of $X$ and $f(y^0, y^1 \mid x)$ be the joint distribution of $(Y^0, Y^1)$ conditional on $X=x$. Then the density of $(Y^0, Y^1, A, X)$ is given by 
\[  f(y^0, y^1, a, x) = f(y^0, y^1  \mid  x)  e(x)^a (1 - e(x))^{1-a} f(x),      \]
and the observed data distribution of $(Y, A, X)$ under Assumption \ref{assumption-1} is  
\[ f(y, a, x) = [ f_1(y  \mid  x) e(x)]^a [ f_0(y \mid x) (1 - e(x))]^{1-a} f(x), \] 
where $f_1(\cdot  \mid x) = \int f(y^0, \cdot  \mid  x)dy^0$ and $f_0(\cdot  \mid  x) = \int f(\cdot, y^1  \mid x)dy^1$ are the marginal distribution function of $Y^1$ and $Y^0$ given $X=x$, respectively.

We first discuss the case where $e(x)$ is unknown. Consider a regular parametric submodel indexed by $\theta$, 
\[  f(y, a, x; \theta)  =  [ f_1(y  \mid x, \theta) e(x, \theta)]^a [ f_0(y \mid x, \theta) (1 - e(x, \theta))]^{1-a} f(x, \theta), \]
which equals $f(y, a, x)$ when $\theta = \theta_0$. Then the score function for the parametric submodel is given by
\begin{align*}
s(y, a, x; \theta) ={}&  a \cdot s_1(y  \mid  x,\theta) + (1 - a) \cdot s_0(y \mid x,\theta) \\
{}&+ \frac{a - e(x,\theta)}{  e(x,\theta) (1- e(x,\theta)) } \dot{e}(x,\theta)    +  s(x, \theta),
\end{align*}	
where $\dot{e}(x, \theta) = \partial e(x, \theta) / \partial \theta$,  $s_1(y  \mid  x, \theta) = \partial \log f_1(y  \mid  x, \theta)/ \partial \theta$, $s_0(y \mid x, \theta) = \partial \log f_0(y \mid x, \theta)/ \partial \theta$,  and  $s(x, \theta) = \partial \log f(x, \theta) / \partial \theta$.  
Thus, the tangent space is 
\begin{align*}
 \mathcal{T} ={}& \Big \{  a \cdot s_1(y \mid x) + (1 - a) \cdot s_0(y \mid x)  +  \alpha(x) \cdot (a - e(x) )   +  s(x): \text{where $s_a(y \mid x) $ satisfies }  \\
  & \quad \text{$\int s_a(y \mid x) f_a(y \mid x)dy = 0$ for $a = 0, 1$, $s(x)$ satisfies $\int s(x) f(x)dx = 0$}, \\
   &\quad \text{and $\alpha(x)$ is an arbitrary  square-integrable measurable function of $x$}\Big\}.        
\end{align*}


Under the parametric submodel indexed by $\theta$, the estimand $\beta_0$ can be written as
\[  	 \beta_0(\theta)  = \int f(x, \theta) \cdot  \int y f_0(y \mid x, \theta) dy \cdot \left ( 1 - \int y f_1(y \mid x, \theta) dy   \right ) dx.    \]
The pathwise derivative of $ \beta_0(\theta)$ at $\theta = \theta_0$ is given as
\begin{align*}
& \frac{\partial \beta_0(\theta)}{ \partial \theta}  \Big |_{\theta = \theta_0}  \\ 
={}&  \int s(x, \theta_0) f(x, \theta_0)   \times \int y f_0(y \mid x, \theta_0) dy  \left ( 1 - \int y f_1(y \mid x, \theta_0) dy  \right ) dx   \\
+{}&  \int f(x, \theta_0)  \times \left\{  \int y s_0(y \mid x,\theta_0) f_0(y \mid x, \theta_0) dy  \left ( 1 - \int y f_1(y \mid x, \theta_0) dy  \right )    \right \} dx  \\
-{}&  	 \int f(x,\theta_0)  \times \left\{  \int y f_0(y \mid x, \theta_0) dy \cdot  \int y s_1(y \mid x, \theta_0) f_1(y \mid x, \theta_0) dy  \}     \right \} dx   \\
={}&  \bfE \left [ s(X) \mu_0(X) (1 - \mu_1(X)) \right ]   +   \bfE \left [  Y^0 s_0(Y^0 \mid  X)     \cdot (1 - \mu_1(X)) \right ]   -   \bfE \left [ \mu_0(X) \cdot  Y^1 s_1(Y^1  \mid  X)  \right ].     
\end{align*}

Next, we verify 
that $\phi(Y,A,X) := \phi_\beta(Y,A,X; \bm{\eta})  - \beta_0$ is the EIF of $\beta_0$. 
We first show that $\phi(Y,A,X) $ is an influence function of $\beta_0$. It suffices to verify that 
\begin{align}
\frac{\partial \beta_0(\theta)}{ \partial \theta}  \Big |_{\theta = \theta_0}  ={}&  \bfE[  \phi(Y, A, X)  \cdot s(Y, A, X;  \theta_0) ] \label{s1}.
\end{align} 
The right-hand side of  (\ref{s1}) can be decomposed as: 
\begin{align*}
\bfE[  \phi(Y, A, X)  \cdot s(Y, A, X;  \theta_0) ]  =  B_1  + B_2 + B_3 + B_4,   
\end{align*}
where 
\begin{align*}
B_1  ={}& \bfE[  \phi(Y,A,X)   \cdot  A \cdot s_1(Y \mid X)  ],   \\
B_2  ={}&  \bfE[  \phi(Y,A,X)   \cdot  (1-A) \cdot s_0(Y \mid X)  ],   \\
B_3  ={}&  \bfE\left [  \phi(Y,A,X)   \cdot   \frac{A - e(X)}{  e(X) (1- e(X)) } \dot{e}(X) \right ],   \\
B_4  ={}&  \bfE[  \phi(Y,A,X)   \cdot   s(X)   ]. 
\end{align*} 
We analyze $B_1, B_2, B_3$, and $B_4$ one by one. Since $\bfE[ s_a(Y \mid X)  \mid  X=x ] = 0$ for $a = 0, 1$, we have 
\begin{align*}
B_1 ={}& \bfE \Big [ \Big \{  - \frac{A (Y- \mu_1(X))}{e(X)} \mu_0(X)  + \mu_0(X)(1-\mu_1(X)) - \beta_0  \Big \}   \cdot A \cdot s_1(Y \mid X) \Big ]  \\
={}&  \bfE \left [  - \frac{A Y}{e(X)} \mu_0(X) s_1(Y \mid X) \right   ]  \\
={}&-   \bfE \left [ \mu_0(X) \cdot  \bfE\{ Y^1 s_1(Y^1  \mid  X)  \mid  X \}  \right ]	\\	
={}& \text{the third term of } \frac{\partial \beta_0(\theta)}{ \partial \theta}  \Big |_{\theta = \theta_0}.  
\end{align*}
Likewise, 
\begin{align*}
B_2  ={}&  \bfE \left [    \frac{(1-A)(Y - \mu_0(X))}{1-e(X)} (1 - \mu_1(X))  \cdot  (1-A) \cdot s_0(Y \mid X) \right ]     \\
={}&\bfE \left [  (1 - \mu_1(X)) \cdot \bfE\{ Y^0 s_0(Y^0 \mid  X)  \mid  X \}    \right ]    \\
={}&  \text{the second term of }  \frac{\partial \beta_0(\theta)}{ \partial \theta}  \Big |_{\theta = \theta_0}.  
\end{align*} 
In addition,
\begin{align*}
B_3  
={}& \bfE \Big [   \frac{(1-A)(Y - \mu_0(X))}{1-e(X)} (1 - \mu_1(X))  \frac{0 -e(X)}{e(X)(1- e(X))} \dot{e}(X) \Big ] \\
{}&- \bfE \Big [  \frac{A (Y- \mu_1(X))}{e(X)} \mu_0(X)   \frac{1- e(X) }{e(X)(1- e(X))} \dot{e}(X)  \Big ]  \\
{}&  + \bfE \Big [ \mu_0(X)(1- \mu_1(X))   \frac{A - e(X)}{  e(X) (1- e(X)) } \dot{e}(X) \Big ] \\
={}& 0, 
\end{align*} 
where the last equation follows from the law of iterated expectations. Also, 
\begin{align*}
B_4  ={}&  \bfE[ \mu_0(X)(1-\mu_1(X)) \cdot   s(X)  ] \\
={}& \text{the first term of } \frac{\partial \beta_0(\theta)}{ \partial \theta}  \Big |_{\theta = \theta_0}.  	\end{align*}  

Combining the results of $B_1$, $B_2$, $B_3$, and $B_4$ leads to  (\ref{s1}).  
In addition, let $s_0(Y \mid X) = \frac{(Y - \mu_0(X))}{1-e(X)} (1 - \mu_1(X))$, $s_1(Y \mid X) = - \frac{ (Y- \mu_1(X))}{e(X)} \mu_0(X)$, and $s(X) =  \mu_0(X)(1-\mu_1(X)) - \beta_0$, then  
                                                                                                                                                       \begin{align*}  \phi(Y,A,X)  ={}&  A \cdot s_1(Y \mid X) + (1- A) s_0(Y \mid X) + s(X),
                                                                                                                                                       \end{align*}
                                                                                                                                                       which implies that $ \phi(Y,A,X)  \in \mathcal{T}$ and thus the EIF of $\beta_0$.

                                               We then discuss the case where the true propensity score $e(x)$ is known. The parametric submodel indexed by $\theta$ is, 
                                                                                                                                                       \[  f(y, a, x; \theta)  =  [ f_1(y\mid x, \theta) e(x) ]^a [ f_0(y\mid x, \theta) (1 - e(x)) ]^{1-a} f(x, \theta), \]
                                                                                                                                                       which equals $f(y, a, x)$ when $\theta = \theta_0$.  The associated score function becomes 
                                                                                                                                                       \begin{align*}
                                                                                                                                                       s(y, a, x; \theta) =  a \cdot s_1(y|x,\theta) + (1 - a) \cdot s_0(y|x,\theta)  +  s(x, \theta),
                                                                                                                                                       \end{align*}	
                                                                                                                                                       Thus, the tangent space is 
                                                                                                                                                       \begin{align*}
                                                                                                                                                        \mathcal{T} ={}&  \Big\{  a \cdot s_1(y\mid x) + (1 - a) \cdot s_0(y\mid x)   +  s(x):   \text{ where $s_a(y\mid x) $ satisfies } \\ & \text{$\int s_a(y\mid x) f_a(y\mid x)dy = 0$ for $a = 0, 1$, and $s(x)$ satisfies $\int s(x) f(x)dx = 0$}  \Big\}.                                                                \end{align*}                                                                                                                                                   
                                                                                                                                                                                                                                              Then by a similar discussion, we can verify that  (\ref{s1}) holds and $ \phi(Y,A,X)  \in \mathcal{T}$. Thus, the EIF of $\beta_0$ remains the same no matter whether the propensity score $e(x)$ is known or not. 

                                                                                                                                                       \bigskip 
                                                                                                                                                       {\bf EIF of $\gamma$}.
                                                                                                                                                       The parametric submodel and the tangent space are the same as in the proof of EIF of $\beta_0$.  
                                                                                                                                                       Under the parametric submodel indexed by $\theta$, the estimand $\gamma$ can be written as
                                                                                                                                                       \begin{align*}
                                                                                                                                                       \gamma(\theta) ={}& \bfE\Big [ \sqrt{  \mu_0(X)(1-\mu_0(X) ) \mu_1(X)(1-\mu_1(X) )    } \Big ) \Big ]  \\
                                                                                                                                                       ={}& \int f(x, \theta) \sqrt{ \int y f_0(y\mid x,\theta)dy  \left \{ 1- \int y f_0(y\mid x,\theta)dy \right \}   }  \\
                                                                                                                                                       {}& \times \sqrt{ \int y f_1(y\mid x,\theta)dy  \left \{ 1- \int y f_1(y\mid x,\theta)dy \right \} } dx. 	
                                                                                                                                                       \end{align*}
                                                                                                                                                    The pathwise derivative of $\gamma(\theta)$ at $\theta = \theta_0$ is 
                                                                                                                                                       \begin{align*}
                                                                                                                                                       & \frac{\partial \gamma(\theta)}{ \partial \theta}   \Big |_{\theta = \theta_0}  \\
                                                                                                                                                       ={}&    \int  s(x) f(x) \sqrt{ m(x) } dx	 \\
                                                                                                                                                       +{}&  	 \int  f(x)   \frac{\mu_1(x)(1-\mu_1(x))}{ \sqrt{m(x)}   }  \frac{(1- 2 \mu_0(x))}{2} \int y s_0(y\mid x) f_0(y\mid x) dydx      \\
                                                                                                                                                       +{}& \int  f(x)   \frac{\mu_0(x)(1-\mu_0(x))}{ \sqrt{ m(x) }   }  \frac{(1- 2 \mu_1(x))}{2} \int y s_1(y\mid x) f_1(y\mid x) dydx,  \\
                                                                                                                                                       ={}&  \bfE[ s(X) \sqrt{m(X)}  ] \\
                                                                                                                                                       {}& + \bfE\left [   \frac{\mu_1(X)(1-\mu_1(X))}{ \sqrt{m(X)}   }  \frac{(1- 2 \mu_0(X))}{2}   \cdot Y^0 s_0(Y^0 \mid X)\right] \\
                                                                                                                                                       {}& + \bfE\left[   \frac{\mu_0(X)(1-\mu_0(X))}{ \sqrt{ m(X) }   }  \frac{(1- 2 \mu_1(X))}{2}  \cdot Y^1 s_1(Y^1\mid X)\right].
                                                                                                                                                       \end{align*}
                                                                                                                                                                                                                                                                                                           Let $\tilde \phi(Y,  A, X)  := \phi_\gamma(Y,A,X; \bm{\eta})  - \gamma$.
Next, we show that                                                                                                                                                         \begin{align}
                                                                                                                                                       \frac{\partial \gamma(\theta)}{ \partial \theta}  \Big |_{\theta = \theta_0}  ={}&  \bfE[ \tilde  \phi(Y, A, X)  \cdot s(Y, A, X;  \theta_0) ]. \label{s3}
                                                                                                                                                       \end{align} 
 The right-hand side of (\ref{s3}) can be decomposed as 
 \[  \bfE[ \tilde  \phi(Y, A, X)  \cdot s(Y, A, X;  \theta_0) ] = B_5 + B_6 + B_7 + B_8, \]
 where                                                                                                                                                   \begin{align*}
     B_5  ={}& \bfE[ \tilde  \phi(Y, A, X)   \cdot  A \cdot s_1(Y\mid X)  ],   \\
  B_6  ={}&  \bfE[  \tilde  \phi(Y, A, X)   \cdot  (1-A) \cdot s_0(Y\mid X)  ],   \\
  B_7  ={}&  \bfE \left [  \tilde  \phi(Y, A, X)    \cdot   \frac{A - e(X)}{  e(X) (1- e(X)) } \dot{e}(X) \right ],   \\
   B_8  ={}&  \bfE[  \tilde  \phi(Y, A, X)    \cdot   s(X)   ]. 
 \end{align*} 
 We can simplify $B_5$ as 
  \begin{align*}
 B_5 ={}& \bfE\left [  \frac{ 1 - 2\mu_1(X) }{2}  \frac{ \mu_0(X) (1-\mu_0(X)) }{ \sqrt{m(X)} }   \frac{A(Y - \mu_1(X))}{e(X)} \cdot s_1(Y\mid X) \right ] \\
      ={}&  \bfE \left [ \frac{ 1 - 2\mu_1(X) }{2}  \frac{ \mu_0(X) (1-\mu_0(X)) }{ \sqrt{m(X)} }   (Y^1 - \mu_1(X)) s_1(Y^1\mid X) \right   ]  \\
             ={}& \bfE \left [ \frac{ 1 - 2\mu_1(X) }{2}  \frac{ \mu_0(X) (1-\mu_0(X)) }{ \sqrt{m(X)} }   Y^1  s_1(Y^1\mid X) \right   ]  \\
                ={}& \text{the third term of } \frac{\partial \gamma(\theta)}{ \partial \theta}  \Big |_{\theta = \theta_0}. 
                  \end{align*}           
 Similarly,                                                                                                                              \begin{align*}
B_6 ={}& \bfE\left [   \frac{ 1 - 2\mu_0(X) }{2} \frac{ \mu_1(X) (1-\mu_1(X)) }{  \sqrt{m(X)} }   \frac{ (1-A)(Y - \mu_0(X))}{1 - e(X)}  \cdot s_0(Y\mid X) \right ] \\
                                                                                                                                                       ={}&  \bfE \left [   \frac{ 1 - 2\mu_0(X) }{2} \frac{ \mu_1(X) (1-\mu_1(X)) }{  \sqrt{m(X)} }  (Y^0 - \mu_0(X)) s_0(Y^0\mid X) \right   ]  \\
                                                                                                                                                       ={}& \bfE \left [ \frac{ 1 - 2\mu_1(X) }{2}     \frac{ 1 - 2\mu_0(X) }{2} \frac{ \mu_1(X) (1-\mu_1(X)) }{  \sqrt{m(X)} } Y^0 s_0(Y^0\mid X)  \right   ]  \\
                                                                                                                                                       ={}& \text{the second term of } \frac{\partial \gamma(\theta)}{ \partial \theta}  \Big |_{\theta = \theta_0}. 
                                                                                                                                                       \end{align*}
                                                                                                                                                       Similar to the analysis of $B_3$, we can verify that 
  $$B_7 = 0.$$ 
     Finally, we have 
\begin{align*}
                 B_8 ={}&  \bfE[ \sqrt{m(X)} \cdot   s(X)  ] \\
   ={}& \text{the first term of } \frac{\partial \gamma(\theta)}{ \partial \theta}  \Big |_{\theta = \theta_0}.  	
                                                                                                                                         \end{align*}
    Combining the results of $B_5$, $B_6$, $B_8$, and $B_8$ leads to  (\ref{s3}).                                                                                                                                                   Also, we can verify that $\tilde  \phi(Y,  A, X) \in \mathcal{T}$, and thus
     $\tilde \phi(Y, A, X)$ is the EIF of $\gamma$.   
                                                     
                                                                                                                                                       In addition, when the true propensity score $e(X)$ is known, we can similarly show                                                                                                                           that $\tilde 
                                                                                                                                                       \phi(Y,  A, X)$ is the EIF of $\gamma$.  
                                                                                                                                                       This completes the proof. 
                                                                                            
                                                                                                                                                       \end{proof}

   \subsection{Proof of Theorem \ref{thm-4}}
                                                                                                                                                       
    \begin{proof}[Proof of Theorem \ref{thm-4}.]  Recall that for $\rho \leq 0$,  we have 
    \begin{align*}  \hat \beta_{\rho} =  \frac{1}{n}  \sum_{i=1}^n   \varphi(Y_i, A_i,X_i; \boldsymbol{\hat  \eta}, \rho), 
\end{align*}
where $\varphi(Y, A, X; \bm{\eta}, \rho) =  \phi_\beta(Y,A,X; \bm{\eta})  -  \rho \cdot  \phi_\gamma(Y,A,X; \bm{\eta}),$  with $ \phi_\beta(Y,A,X; \boldsymbol{\eta})$ and $\phi_\gamma(Y,A,X; \bm{\eta})$ defined in Theorem \ref{thm-3}.
              Therefore, $n^{1/2} (\hat \beta_\rho - \beta_\rho) = n^{1/2} (\hat \beta_0 - \beta_0 ) - \rho \cdot n^{1/2} ( \hat \gamma - \gamma),$
    where $\hat \beta_{0} =  n^{-1}  \sum_{i=1}^n   \phi_\beta(Y_i, A_i,X_i; \boldsymbol{\hat  \eta})$ and  $\hat \gamma =   n^{-1}   \sum_{i=1}^n   \phi_\gamma(Y_i, A_i,X_i; \boldsymbol{\hat  \eta})$.

 We first analyze $n^{1/2} (\hat \beta_0 - \beta_0)$ and then the term $n^{1/2} ( \hat \gamma - \gamma)$ can be addressed similarly.
 We decompose $n^{1/2} (\hat \beta_0 - \beta_0)$ as follows: 
                                                                    \begin{align*} n^{1/2} (\hat \beta_0 - \beta_0) 
                                                                                                                                                       ={}&   H_{1n} + H_{2n} + H_{3n}, 
                                                                                                                                                       \end{align*} 
                                                                                                                                                       where 
                                                                                                                                                       \begin{align*}
                                                                                                                                                       H_{1n} ={}&  n^{-1/2}  \sum_{i=1}^n \Big \{ \phi_\beta(Y_i,A_i,X_i; \boldsymbol{\eta})  - \bfE[   \phi_\beta(Y_i,A_i,X_i; \boldsymbol{\eta})]    \Big \}, \\
                                                                                                                                                       H_{2n} ={}&  n^{-1/2} \sum_{i=1}^n \Big \{ \phi_\beta(Y_i,A_i,X_i; \boldsymbol{\hat \eta})  - \phi_\beta(Y_i,A_i,X_i; \boldsymbol{\eta}) - \bfE \left [ \phi_\beta(Y,A,X; \boldsymbol{\hat \eta})  -\phi_\beta(Y,A,X; \boldsymbol{\eta}) \right ]  \Big \},   \\
                                                                                                                                                       H_{3n} ={}&   n^{1/2} \bfE \left [ \phi_\beta(Y,A,X; \boldsymbol{\hat \eta})  -\phi_\beta(Y,A,X; \boldsymbol{\eta}) \right ]. 
                                                                                                                                                       \end{align*}
                                                                                                                                                       By the central limit theorem for  independent and identically distributed random variables,   \[  H_{1n}     \xrightarrow{d} \mathcal{N}(0,  \sigma^2),  \]
                                                                                                                                                       where $\sigma^2 = \mathbb{V}[ \phi_{\beta}(Y, A, X; \boldsymbol{\eta}) ]$. Then, we show that $H_{2n} = o_{\P}(1)$ and $H_{3n} = o_{\P}(1)$. 
                                                                                                                                                       
                                                                                                                                                         For $H_{2n}$, due to sample splitting, applying Chebyshev's inequality gives  
    \[   H_{2n}  = O_{\P}( || \phi_\beta(Y,A,X; \boldsymbol{\hat \eta})  -\phi_\beta(Y,A,X; \boldsymbol{\eta})    ||_2 ).  \]
Because $\mu_a(X), 1/(1-e(X)), 1/e(X)$ are all bounded and $\phi_\beta(Y,A,X; \boldsymbol{\eta})$ is a Lipschitz function of $\boldsymbol{\eta}$,  by Condition \ref{con-1}, we have that 
    \[ || \phi_\beta(Y,A,X; \boldsymbol{\hat \eta})  -\phi_\beta(Y,A,X; \boldsymbol{\eta})    ||_2 = o_{\P}(1), \]
  which implies that $H_{2n} = o_{\P}(1).$ 

Next, we focus on $H_{3n}$. Define the Gateaux derivative of the generic function $g$ in the direction $[\hat e - e,  \hat \mu_0 - \mu_0, \hat \mu_1 - \mu_1]$ by 
 $\partial_{[\hat e - e,  \hat \mu_0 - \mu_0, \hat \mu_1 - \mu_1]}g$.  
 By a Taylor expansion, we have 
   \begin{align*}
H_{3n} ={}&  n^{1/2} \cdot \bfE \left [ \phi_\beta(Y,A,X; \boldsymbol{\hat \eta})  - \phi_\beta(Y,A,X; \boldsymbol{\eta}) \right ]   \\
     ={}&n^{1/2} \cdot   \partial_{[\hat e - e,  \hat \mu_0 - \mu_0, \hat \mu_1 - \mu_1]} \bfE[  \phi_\beta(Y,A,X; \boldsymbol{\eta})   ] \\
      &+ n^{1/2} \cdot  \frac 1 2  \partial^2_{[\hat e - e,  \hat \mu_0 - \mu_0, \hat \mu_1 - \mu_1]} \bfE[  \phi_\beta(Y,A,X; \boldsymbol{\eta})   ] +  R( (\boldsymbol{\hat \eta} - \boldsymbol{\eta}) ),
   \end{align*}
where  $R( (\boldsymbol{\hat \eta} - \boldsymbol{\eta}) )$ represents the remainder term. 
The first-order term equals zero:  
	\begin{align*}
	& n^{1/2} \cdot  \partial_{[\hat e - e,  \hat \mu_0 - \mu_0, \hat \mu_1 - \mu_1]} \bfE[  \phi_\beta(Y,A,X; \boldsymbol{\eta})   ]  \\
	={}&  n^{1/2} \cdot  \bfE \biggl [   \left \{  \frac{  (1-A) \{ Y - \mu_0(X)\} }{ (1-e(X))^2 }  - \frac{  A\{ Y - \mu_1(X)\} }{ e(X)^2 }   \right \}  \times \{\hat e(X) - e(X) \}             \\
	+{}&    \left \{   - \frac{ (1 - A) (1-\mu_1(X))  }{  1 - e(X) } - \frac{A(Y -\mu_1(X)) }{e(X)}  + (1-\mu_1(X) )  \right \}  \times \{  \hat \mu_0(X) - \mu_0(X) \}     \\
	+{}&     \left \{ -\frac{ (1-A)(Y - \mu_0(X)) }{1 - e(X)}  + \frac{A\mu_0(X)}{e(X)}  - \mu_0(X) \right \}  \times \{  \hat \mu_1(X) - \mu_1(X) \}     \biggr ]   \\
	={}& 0,
	\end{align*}
where the last equation follows from the sample splitting, $\bfE[ A(Y -\mu_1(X)) g(X) | X] = 0$, $\bfE[ (1- A)(Y -\mu_0(X)) g(X) | X] = 0$, $\bfE[ A g(X) / e(X) | X] = g(X), \bfE[ (1-A) g(X)/(1-e(X))|X] = g(X)$, where $g$ is an arbitrary integrable function of $X$.   

 For the second-order term, some calculations yield that 
	\begin{align*}
	 & \frac 1 2  n^{1/2} \cdot   \partial^2_{[\hat e - e,  \hat \mu_0 - \mu_0, \hat \mu_1 - \mu_1]} \bfE[  \phi_\beta(Y,A,X; \boldsymbol{\eta})   ]  \\
	 ={}& \frac 1 2  n^{1/2} \cdot  \bfE \biggl [   \left \{  \frac{ 2 (1-A) \{ Y - \mu_0(X)\} }{ (1-e(X))^3 }  + \frac{  2 A\{ Y - \mu_1(X)\} }{ e(X)^2 }   \right \}  \times \{\hat e(X) - e(X) \}^2               \\
	 +{}&   \left \{  \frac{  -(1-A) \{ \hat \mu_0(X) - \mu_0(X)\} }{ (1-e(X))^2 }  + \frac{ A\{ \hat \mu_1(X) - \mu_1(X)\} }{ e(X)^2 }   \right \}  \times \{\hat e(X) - e(X) \} \\
	+{}&   \left \{  \left ( - \frac{ (1 - A) (1-\mu_1(X))  }{  (1 - e(X))^2 } + \frac{A(Y -\mu_1(X)) }{e(X)^2}  \right) \{\hat e(X) - e(X)  \}  \right \}  \times \{  \hat \mu_0(X) - \mu_0(X) \}     \\
	+{}&   \left \{   \left (  \frac{ (1 - A)  }{  1 - e(X) } +  \frac{A}{e(X)}   -1 \right )\cdot \{ \hat \mu_1(X) -\mu_1(X) \}   \right \}  \times \{  \hat \mu_0(X) - \mu_0(X) \}     \\ 
	+{}&     \left \{ \left(  -\frac{ (1-A)(Y - \mu_0(X)) }{ (1 - e(X))^2 }  + \frac{A\mu_0(X)}{e(X)^2}  \right ) \{\hat e(X) - e(X) \} \right \}  \times \{  \hat \mu_1(X) - \mu_1(X) \}     \\	
		+{}&    \left \{ \left ( \frac{ (1-A) }{1 - e(X)}  + \frac{A}{e(X)}  - 1 \right ) \{\hat \mu_0(X) - \mu_0(X)  \} \right \}  \times \{  \hat \mu_1(X) - \mu_1(X) \}     \biggr ]   \\	
	={}&  n^{1/2} \cdot  O_{\P} \left (  || \hat e(X) - e(X)  ||_2 \cdot    || \hat \mu_1(X) - \mu_1(X) ||_2     \right )  \\
  {}& n^{1/2} \cdot  O_{\P} \left (  || \hat e(X) - e(X)  ||_2 \cdot   || \hat \mu_0(X) - \mu_0(X) ||_2     \right )  \\ 
 {}& + n^{1/2} \cdot  O_{\P} \left (   || \hat \mu_1(X) - \mu_1(X) ||_2 \cdot  || \hat \mu_0(X) - \mu_0(X) ||_2  \right ) \\
	={}& o_{\P}(1),
	\end{align*}
where the last equality follows from Condition \ref{con-1}. 
The remainder term is dominated by the second-order term. 
Thus, we have that 
  $$n^{1/2} (\hat \beta_0 - \beta_0) = n^{-1/2}  \sum_{i=1}^n \Big \{ \phi_\beta(Y_i,A_i,X_i; \boldsymbol{\eta})  - \bfE[   \phi_\beta(Y_i,A_i,X_i; \boldsymbol{\eta})]    \Big \} + o_{\P}(1).$$
Similarly, 
    $$ n^{1/2} ( \hat \gamma - \gamma) = n^{-1/2}  \sum_{i=1}^n \Big \{ \phi_\gamma(Y_i,A_i,X_i; \boldsymbol{\eta})  - \bfE[   \phi_\gamma(Y_i,A_i,X_i; \boldsymbol{\eta})]    \Big \} + o_{\P}(1). $$
This completes the proof. 

\end{proof}

 
\subsection{Proof of Theorem \ref{thm-5}} 
Let $\lesssim$ represent smaller than up to a constant.  
    The following Lemmas \ref{lem-s1} and \ref{lem-s2} will be used in the proof of Theorem \ref{thm-5}, where Lemma \ref{lem-s1} can be found 
 in the proof of Theorem 2 in \cite{Bonvini-Kennedy2022}. 
 We present it here for completeness. 

\begin{lemma} \label{lem-s1} Let $\hat f$ and $f$ take any real values. Then
		\[     \big |  \mathbb{I}(\hat  f > 0) -  \mathbb{I}(f > 0 ) \big | \leq  \mathbb{I}( |f| \leq |\hat f - f| ).        \]
\end{lemma}
\begin{proof}[Proof of Lemma \ref{lem-s1}.]   Note that 
	\[  \big |  \mathbb{I}(\hat  f > 0) -  \mathbb{I}(f > 0 ) \big | = \mathbb{I}(\hat f, f \text{ have opposite signs}).     \]
 If $\hat f$ and $f$ have the opposite signs, then 
	\[   |\hat f|  + |f|  = |\hat f - f |,     \]
which implies that $|f| \leq |\hat f - f| $. Thus, whenever $   \big |  \mathbb{I}(\hat  f > 0) -  \mathbb{I}(f > 0 ) \big | = 1$, it must have $\mathbb{I}( |f| \leq |\hat f - f| ) = 1$. This completes the proof. 

\end{proof}



\medskip
\begin{lemma} \label{lem-s2}  Under Assumptions 1, we have     
	\begin{align*} 
	 & \bfE[  \phi_\beta(Y,A,X; \boldsymbol{\hat \eta}) - \phi_\beta(Y,A,X; \boldsymbol{ \eta}) \mid X=x ]  \\
	 \lesssim{}&   (\hat e(x) - e(x))\Big \{  (\hat \mu_1(x) - \mu_1(x)) +  (\hat \mu_0(x) - \mu_0(x)) \Big \}     + (\hat \mu_1(x) - \mu_1(x))  (\hat \mu_0(x) - \mu_0(x))  
	 \end{align*}
and 
\begin{align*} 
	 & \bfE[  \phi_\gamma(Y,A,X; \boldsymbol{\hat \eta}) -  \phi_\gamma(Y,A,X; \boldsymbol{ \eta}) \mid X=x ]  \\
	 \lesssim{}&   (\hat e(x) - e(x))\Big \{  (\hat \mu_1(x) - \mu_1(x)) +  (\hat \mu_0(x) - \mu_0(x)) \Big \}   \\
		 {}&  + (\hat \mu_1(x) - \mu_1(x))  (\hat \mu_0(x) - \mu_0(x))  + (\hat \mu_1(x) - \mu_1(x))^2 +   (\hat \mu_0(x) - \mu_0(x))^2.  
	 \end{align*}
\end{lemma}
\begin{proof}[Proof of Lemma \ref{lem-s2}.]  
This follows by a similar argument of $H_{3n}$ in the proof of Theorem \ref{thm-4}.  

\end{proof}

\medskip 
Next, we  prove Theorem \ref{thm-5}. 

\begin{proof}[Proof of Theorem \ref{thm-5}.]  
  Similar to the decomposition of $n^{1/2} (\hat \beta_{0} - \beta_{0})$, we can decompose $n^{1/2} (\hat \beta_{\rho} - \beta_{\rho})$ for $\rho > 0$  as
	\begin{align*} n^{1/2} (\hat \beta_{\rho} - \beta_{\rho}) 
	={}&   H_{4n} + H_{5n} + H_{6n}, 
	\end{align*} 
where 
	\begin{align*}
		H_{4n} ={}&  n^{-1/2}  \sum_{i=1}^n \Big \{ \varphi(Y_i,A_i,X_i; \boldsymbol{\eta})  - \bfE[   \varphi(Y_i,A_i,X_i; \boldsymbol{\eta})]    \Big \}, \\
		H_{5n} ={}&  n^{-1/2} \sum_{i=1}^n \Big \{ \varphi(Y_i,A_i,X_i; \boldsymbol{\hat \eta})  - \varphi(Y_i,A_i,X_i; \boldsymbol{\eta})   - \bfE \left [ \varphi(Y_i,A_i,X_i; \boldsymbol{\hat \eta})  -\varphi(Y_i,A_i,X_i; \boldsymbol{\eta}) \right ]  \Big \},   \\
	H_{6n} ={}&   n^{1/2} \bfE \left [ \varphi(Y,A,X; \boldsymbol{\hat \eta})  -\varphi(Y,A,X; \boldsymbol{\eta}) \right ]. 
	\end{align*}
  By the central limit theorem for independent and identically distributed random variables,
   \[  H_{4n}     \xrightarrow{d} N(0,  \sigma_{\rho}^2),  \]
  where $\sigma_{\rho}^2 = \mathbb{V}[ \varphi(Y, A, X; \boldsymbol{\eta}) ]$.  Then, it suffices to show that $H_{5n} = o_{\P}(1)$ and $H_{6n} = o_{\P}(1)$.

By a similar argument of $H_{2n}$ in the proof of Theorem \ref{thm-4},  we have  $H_{5n} = o_{\P}(1).$  
Next, we discuss $H_{6n}$. Different from $H_{3n}$, $H_{6n}$ is not a continuous function of $\boldsymbol{\eta}$ and consequently, we cannot resort to the Taylor expansion in terms of the Gateaux derivative. To proceed, we further decompose $H_{6n}$ as 
	\[       H_{6n} = H_{6n,1} + H_{6n,2},      \]
where 
	\begin{align*}
		H_{6n, 1} ={}&  n^{1/2} \bfE \Big [  \mathbb{I}\{g(\boldsymbol{\hat \eta}, \rho) \geq 0 \} \cdot \Big \{ \Big ( \phi_\beta(Y,A,X; \boldsymbol{\hat \eta})  -  \rho \cdot  \phi_\gamma(Y,A,X; \boldsymbol{\hat \eta}) \Big ) \\
		 {}& - \Big ( \phi_\beta(Y,A,X; \boldsymbol{\eta})  -  \rho \cdot  \phi_\gamma(Y,A,X; \boldsymbol{\eta}) \Big )   \Big \}  \Big ], \\  
		H_{6n, 2} ={}&  n^{1/2} \bfE \Big [  \Big \{ \mathbb{I}(g(\boldsymbol{\hat \eta}, \rho) \geq 0) -  \mathbb{I}(g(\boldsymbol{ \eta}, \rho) \geq 0)  \Big \} \\
  {}& \times \Big \{ \phi_\beta(Y,A,X; \boldsymbol{\eta})  -  \rho \cdot  \phi_\gamma(Y,A,X; \boldsymbol{\eta})  \Big \} \Big ]. 
	\end{align*}

Due to sample splitting, we can take $\boldsymbol{\hat \eta}$ as fixed function of $X$ without loss of generality, then by Lemma \ref{lem-s2} and the Cauchy–Schwarz inequality,    
	 \begin{align*}
	 	H_{6n, 1} &={} n^{1/2}  \bfE \Big [    \mathbb{I}\{g(\boldsymbol{\hat \eta}, \rho) \geq 0 \} \\
   {}& \times \bfE \Big \{ \Big ( \phi_\beta(Y,A,X; \boldsymbol{\hat \eta})  -  \rho \cdot  \phi_\gamma(Y,A,X; \boldsymbol{\hat \eta}) \Big )  - \Big ( \phi_\beta(Y,A,X; \boldsymbol{\eta})  -  \rho \cdot  \phi_\gamma(Y,A,X; \boldsymbol{\eta}) \Big ) \Big | X  \Big \}     \Big ] \\
		\lesssim{}&     n^{1/2}  \bfE \Big [    \mathbb{I}\{g(\boldsymbol{\hat \eta}, \rho) \geq 0 \} \cdot  \Big \{ (\hat e(X) - e(X)) (\hat \mu_1(X) - \mu_1(X)) \\
  {}& + (\hat e(X) - e(X)) (\hat \mu_0(X) - \mu_0(X))    + (\hat \mu_1(X) - \mu_1(X))  (\hat \mu_0(X) - \mu_0(X)) \\
   {}&+ (\hat \mu_1(X) - \mu_1(X))^2 +   (\hat \mu_0(X) - \mu_0(X))^2  \Big \}  \Big ]  \\
		 \leq{}& n^{1/2}  \bfE^{1/2} \Big [ \mathbb{I}\{g(\boldsymbol{\hat \eta}, \rho) \geq 0 \} \Big ] \\
   {}& \times \bfE^{1/2} \Big [ \Big \{ (\hat e(X) - e(X)) (\hat \mu_1(X) - \mu_1(X))  \\
   {}&+ (\hat e(X) - e(X)) (\hat \mu_0(X) - \mu_0(X))     + (\hat \mu_1(X) - \mu_1(X))  (\hat \mu_0(X) - \mu_0(X)) \\
   {}&+ (\hat \mu_1(X) - \mu_1(X))^2 +   (\hat \mu_0(X) - \mu_0(X))^2  \Big \}^2  \Big ]  \\
		 	 \lesssim{}&  n^{1/2} O_{\P} \Big (  || \hat e(X) - e(X)  ||_2  \times \{ || \hat \mu_1(X) - \mu_1(X) ||_2  +  || \hat \mu_0(X) - \mu_0(X) ||_2  \}    \Big )  \\ 
{}& + n^{1/2} \cdot  O_{\P} \Big (   || \hat \mu_1(X) - \mu_1(X) ||_2  || \hat \mu_0(X) - \mu_0(X) ||_2  \Big ) \\
{}& + n^{1/2} \cdot  O_{\P} \left (   || \hat \mu_1(X) - \mu_1(X) ||_2^2  +  || \hat \mu_0(X) - \mu_0(X) ||_2^2  \right ),
	 \end{align*}
which leads to $H_{6n,1} =o_{\P}(1)$ by Conditions \ref{con-1} and \ref{con-3}.

For $H_{6n, 2}$,  we have   	
	\begin{align*}
		 H_{6n, 2} 
		={}& n^{1/2} \bfE \Big [  \Big \{ \mathbb{I}(g(\boldsymbol{\hat \eta}, \rho) \geq 0) -  \mathbb{I}(g(\boldsymbol{ \eta}, \rho) \geq 0)  \Big \} \cdot g(\boldsymbol{\eta}, \rho) \Big ] \\
     {}& + n^{1/2} \bfE \Big [  \Big \{ \mathbb{I}(g(\boldsymbol{\hat \eta}, \rho) \geq 0) -  \mathbb{I}(g(\boldsymbol{ \eta}, \rho) \geq 0)  \Big \}  \times \Big \{ \phi_\beta(Y,A,X; \boldsymbol{\eta})  -  \rho \cdot  \phi_\gamma(Y,A,X; \boldsymbol{\eta}) - g(\boldsymbol{\eta}, \rho)  \Big \} \Big ] \\
  ={}& n^{1/2} \bfE \Big [  \Big \{ \mathbb{I}(g(\boldsymbol{\hat \eta}, \rho) \geq 0) -  \mathbb{I}(g(\boldsymbol{ \eta}, \rho) \geq 0)  \Big \} \cdot g(\boldsymbol{\eta}, \rho) \Big ] \\
		\lesssim{}& n^{1/2} \bfE \Big [  \Big | \mathbb{I}(g(\boldsymbol{\hat \eta}, \rho) \geq 0) -  \mathbb{I}(g(\boldsymbol{ \eta}, \rho) \geq 0)  \Big |  \cdot \Big | g(\boldsymbol{\eta}, \rho) \Big|  \Big ]   \\
  	\lesssim{}& n^{1/2} \cdot \bfE \left [ \mathbb{I}\Big ( \big | g(\boldsymbol{\eta}, \rho) \big | \leq  \big | g(\boldsymbol{\hat \eta}, \rho) -  g(\boldsymbol{ \eta}, \rho)  \Big )  \big | \cdot  \Big | g(\boldsymbol{\eta}, \rho) \Big|     \right ]  \\
		\lesssim{}& n^{1/2} \cdot || g(\boldsymbol{\hat \eta}, \rho) -  g(\boldsymbol{ \eta}, \rho) ||_{\infty} \cdot \bfE \left [ \mathbb{I}\Big ( \big | g(\boldsymbol{\eta}, \rho) \big | \leq  \big | g(\boldsymbol{\hat \eta}, \rho) -  g(\boldsymbol{ \eta}, \rho)  \big |   \Big )   \right ]  \\
		\lesssim{}& n^{1/2} || g(\boldsymbol{\hat \eta}, \rho) -  g(\boldsymbol{ \eta}, \rho) ||_{\infty}^{1 + \alpha} \\
		={}&  o_{\P}(1). 
	\end{align*} 
where the second equality follows from  $\bfE[ \phi_\beta(Y,A,X; \boldsymbol{\eta})  -  \rho \cdot  \phi_\gamma(Y,A,X; \boldsymbol{\eta})  \mid X ] = g(\boldsymbol{\eta}, \rho)$, the second 
inequality follows from Lemma \ref{lem-s1}, and the fourth inequality holds by Condition \ref{con-2} (the margin condition), and the last equality holds by Condition \ref{con-3}(b).  

\end{proof}


\subsection{Proof of Proposition \ref{prop-10}}
\begin{proof}[Proposition \ref{prop-10}]
We first prove Proposition \ref{prop-10}(a), and then  prove Proposition \ref{prop-10}(b). 

\medskip
\emph{Proposition \ref{prop-10}(a)}.  
By the definition of $\rho(x)$, under Assumption \ref{assumption-1}, we have 
    \begin{align*}
        \rho(x) ={}& \frac{ \E[ Y^0 Y^1 \mid X =x ] - \E[ Y^0 \mid X=x ] \cdot  \E[ Y^1 \mid X=x ] } {\sqrt{ \V(Y^0 \mid X=x) \cdot  \V(Y^1 \mid X=x) } } \\
         ={}&    \frac{\P(Y^0=1, Y^1=1 \mid X=x) - \mu_0(x)\mu_1(x) }{\sqrt{  \mu_0(x)(1-\mu_0(x) ) \mu_1(x)(1-\mu_1(x)  }}. 
    \end{align*} 
 Then Assumption \ref{assumption-2} implies that 
  \begin{equation}    \label{eq-s4}
      \mu_0(x)  \mu_1(x) +  \rho_l(x) \sqrt{m(x)}  \leq \pi_{11}(x) = \P(Y^0=1, Y^1=1 \mid X=x)  \leq  \mu_0(x)  \mu_1(x) +  \rho_u(x) \sqrt{m(x)}.    
     \end{equation} 
By a similar proof of Theorem \ref{thm-1}, 
     \[       \mu_0(x)  \mu_1(x) +  \rho_u(x) \sqrt{m(x)} \leq 1,  \]
and we must have $ \mu_0(x)  \mu_1(x) +  \rho_u(x) \sqrt{m(x)} $ is larger than zero as $\pi_{11}(x)$ is a probability.  
Furthermore, by the Fr\'{e}chet--Hoeffding inequality, 
       \begin{align*}
           \pi_{11}(x) =  \P(Y^0=1, Y^1=1\mid X=x) \geq{}& \max\{\mu_0(x) + \mu_1(x) -1, 0 \} \\
           \pi_{11}(x)  =   \P(Y^0=1, Y^1=1\mid X=x) \leq{}& \min\{\mu_0(x),  \mu_1(x)\},
       \end{align*}  
This leads to the sharp bounds on $\pi_{11}(x)$. 
The sharp bounds on $\pi_{01}(x)$ and $ \pi_{00}(x)$ follows immediately by noting that 
 \begin{align*}
     \pi_{01}(x) ={}& \P(Y^1 =1\mid X=x) - \P(Y^0=1, Y^1=1 \mid X=x) \\ 
     ={}& \mu_1(x) - \P(Y^0=1, Y^1=1\mid X=x),\\
     \pi_{00}(x) ={}& \P(Y^0=0\mid X=x) - \P(Y^0=0, Y^1 = 1\mid X=x) \\
         ={}& (1-\mu_0(x)) - \pi_{01}(x) \\
         ={}& (1-\mu_0(x)) - \mu_1(x) + \P(Y^0=1, Y^1=1 \mid X=x). \end{align*}

  \medskip
\emph{Proposition \ref{prop-10}(b)}.  
Under additional Assumption \ref{assumption-3}, we can verify that 
         \[       \mu_0(x)  \mu_1(x) +  \rho_l(x) \sqrt{m(x)} \geq 0,    \]
           \[       \mu_0(x)  \mu_1(x) +  \rho_u(x) \sqrt{m(x)} \leq  \min\{\mu_0(x),  \mu_1(x)\},    \]
and 
  \[         \mu_0(x)  \mu_1(x) +  \rho_l(x) \sqrt{m(x)}  \geq \mu_0(x) + \mu_1(x) -1.  \]
This implies the bounds on $\pi_{11}(x)$. Similarly, the bounds on $\pi_{01}(x)$ and $\pi_{00}(x)$ can be obtained. 

\end{proof}

\end{document}